%% file: ex_article.tex
\documentclass[preprint,onefignum,onetabnum]{siamart171218}

\input{ex_shared}

\ifpdf
\hypersetup{
  pdftitle={Multilevel Optimal Transport: a Fast Approximation of Wasserstein-1 distances},
  pdfauthor={Jialin Liu, Wotao Yin, Wuchen Li and Yat Tin Chow}
}
\fi

\begin{document}

\maketitle

\begin{abstract}
  We propose a fast algorithm for the calculation of the Wasserstein-1 distance, which is a particular type of optimal transport distance with homogeneous of degree one transport cost.
  Our algorithm is built on multilevel primal-dual algorithms. Several numerical examples and a complexity analysis are provided to demonstrate its computational speed. On some commonly used image examples of size $512\times512$, the proposed algorithm gives solutions within $0.2\sim 1.5$ seconds on a single CPU, which is much faster than the state-of-the-art algorithms.
\end{abstract}



\section{Introduction}
Optimal transport (OT) plays crucial roles in many areas, including fluid dynamics \cite{villani2008optimal}, image processing \cite{rubner2000earth, Ryu}, machine learning \cite{wgan, wloss} and control~\cite{chen2016relation,chen2017optimal}. It is a well-posed distance between two probability distributions over a given domain. The distance is often named Earth Mover's distance (EMD) or the Wasserstein distance. Plenty of theories on OT have been introduced \cite{beckmann1952continuous,benamou2000computational,gangbo1996geometry,otto2000generalization,villani2008optimal,berman2018convergence}.
Despite the theoretical development, computing the distance is still challenging since the OT problems usually do not have closed-form solutions. Fast numerical algorithms are essential for the related applications. 

Recently, a particular class of OT, named the Wasserstein-1 distance, has been widely used in machine learning problems \cite{wgan,gulrajani2017improved,petzka2017regularization}. It gains rising interest in the computational mathematics community \cite{schmitzer2017dynamic,schmitzer2017framework,hartmann2017semi,li2018parallel,bassetti2018computation} ~\cite[Section 3.1]{Solomon}. 
The Wasserstein-1 distance is named as its transport cost is homogeneous of degree one. In this paper, we focus on numerically computing Wasserstein-1 distances. 

In the literature, many numerical schemes have been proposed for the OT problem.
\cite{hillier1995introduction,rubner2000earth,Ling, Pele,pele2008linear,oberman2015efficient,bartels2017error,bassetti2018computation} modeled the OT problem as linear programming (LP) with specific structures. They utilized these structures to develop efficient solvers. \cite{Peyre,PC,Benamou1,haber_horesh_2015,li2018parallel,jacobs2018solving,ryu2018unbalanced} modeled OT as a nonsmooth convex optimization problem and introduced iterative algorithms to solve it.
\cite{cuturi2013sinkhorn,benamou2015iterative,solomon2015convolutional,essid2018quadratically,carlier2017convergence,ferradans2014regularized,cuturi2016smoothed,berman2017sinkhorn} studied the OT problems with regularizers and proposed efficient algorithms to solve them. In particular, some algorithms have been developed for calculating the Wasserstein-1 distance and its variants. 
 Ling and Okada \cite{Ling} exploited the structure of the problem to improve the transportation simplex algorithm~\cite{hillier1995introduction} and proposed Tree-EMD. Pele and Werman \cite{pele2008linear, Pele} proposed and solved EMD with a thresholded ground metric. 
 Bartel and Sch{\"o}n \cite{bartels2017adaptive} added a regularization term on the original problem and introduced primal-dual algorithms to solve it.
 Li et al. \cite{li2018parallel} studied a primal-dual algorithm for calculating Wasserstein-1 distances that is friendly to parallel programming and has an implementation on CUDA. Jacobs et al. \cite{jacobs2018solving} introduced the proximal PDHG method, whose number of iterations is independent of the grid size. 
Bassetti et al. \cite{bassetti2018computation} studied the connections between the Wasserstein-1 distance and the uncapacitated minimum cost flow problem and applied the network simplex algorithm to solve it. 

\paragraph{Motivations and our contributions} Although many numerical algorithms \cite{Ling,bassetti2018computation,li2018parallel,jacobs2018solving} have been proposed to calculate the Wasserstein-1 distance, there is still some room to speed them up, especially for large-scale problems, for example, a grid of $512\times 512$. Motivated by the success of multigrid methods~\cite{wesseling1995introduction} for calculating Wasserstein-$p~(p>1)$ distance~\cite{oberman2015efficient,haber_horesh_2015,schmitzer2016sparse,merigot2011multiscale},
we apply the cascadic multilevel method \cite{bornemann1996cascadic} to calculate Wasserstein-1 distances. We compute the distances on different grid levels and use the solutions on the coarse grids to initialize the calculation of solutions on the finer grids. We use this method to speed up the state-of-the-art algorithms \cite{li2018parallel,jacobs2018solving}, dramatically reducing the computational expense on the finest grids and lessening the total time consumption by $2\sim12$ times. The speedup effect depends on the size of the problem. It is significant for large-scale problems.

The rest of this paper is organized as follows. In Section \ref{section:2}, we briefly review the Wasserstein-1 distance. In Section \ref{section:3}, we demonstrate our multilevel algorithms and provide a complexity analysis in Section \ref{section:4}. In Section \ref{section:assumes}, we numerically validate the assumptions used in Section \ref{section:4}. Finally, in Section \ref{section:5}, we present several numerical examples.  

\section{Problem description}\label{section:2}
Given a domain $\Omega \subset \Re^d$, the EMD, or the Wasserstein distance, is a commonly-used metric to measure the distance between two probability distributions defined on $\Omega$: $\rho^0,\rho^1: \Omega \to \Re$. Wasserstein$-1$ distance is defined by the minimum value of the following minimization problem:
\begin{equation}
\label{eq:opt_map}
\begin{aligned}
W_1(\rho^0,\rho^1) = \inf_{\pi \in \Gamma(\rho^0,\rho^1) }&\int_{\Omega\times \Omega}\|x-y\|_{p} \mathrm{d}\pi(x,y),
\end{aligned}
\end{equation}
where $\|x-y \|_p, 1 \leq p \leq \infty$, defines the tranport cost between two points $x,y\in \Omega$, and $\Gamma(\rho^0,\rho^1)$ is the set of nonnegative measurements $\pi$ on $\Omega\times\Omega$ satisfying
\[\pi(B\times\Omega) = \rho^0(B),\quad\pi(\Omega\times B) = \rho^1(B)\]for all measurable $B\subset \Omega$.
The dual problem of (\ref{eq:opt_map}), also named the Kantorovich dual problem, is:
\begin{equation}
\label{eq:opt_map_dual}
W_1(\rho^0,\rho^1) = \sup\bigg\{\int_{\Omega}\phi^0\mathrm{d}\rho^0 - \phi^1\mathrm{d}\rho^1 \bigg|  \phi^0(x) - \phi^1(y) \leq \|x-y\|_{p},~ \forall x,y\in\Omega.\bigg\},
\end{equation}
where $\phi^0,\phi^1: \Omega \to \Re,$ are (Kantorovich) dual variables.

In the 1D case ($d=1$), the Wasserstein Distance has a  closed-form solution \cite{villani2008optimal}. With two or higher dimensions ($d \geq 2$), the distance is no longer given in a closed form, and it is obtained via iterative algorithms.

\subsection{Problem settings}In this paper, we focus on an equivalent and simpler form of (\ref{eq:opt_map_dual}). Since $\|\cdot\|_p$ is homogeneous of degree one, by \cite{villani2008optimal}, there is an equivalent form of (\ref{eq:opt_map_dual}), where $\phi^0=\phi^1=\phi$. In other words,
\begin{equation}
\label{eq:opt_map_dual2}
\begin{aligned}
\maximize_{\phi: \Omega \to \Re}&\int_{x\in\Omega}\phi(x)(\rho^0(x) - \rho^1(x))dx\\
\text{subject to }&\|\nabla \phi(x)\|_{q} \leq 1,\quad \text{almost everywhere in } \Omega,
\end{aligned}
\end{equation}
where $1/p + 1/q = 1$ and $1 \leq q \leq \infty$.
The following minimization problem, which is the dual problem of (\ref{eq:opt_map_dual2}), is also considered in this paper:
\begin{equation}
\label{eq:opt_flux}
\begin{aligned}
\minimize_{m: \Omega \to \Re^d}& \int_{x\in\Omega}  \|m(x)\|_{p} dx \\
\text{subject to }& \text{div} (m(x)) = \rho^0(x) - \rho^1(x),\quad \forall x\in\Omega, \\
& m(x)\cdot n(x) = 0,\quad  \forall x \in \partial \Omega,
\end{aligned}
\end{equation}
where ``div'' denotes the divergence operator $\text{div} (m(x)) = \sum_{i=1}^d \frac{\partial m_i}{\partial x_i}(x)$ and $n(x)$ is normal to $\partial \Omega$. 
Here $m$ is a $d$ dimensional field satisfying the zero flux boundary condition~\cite{beckmann1952continuous}. The solution of (\ref{eq:opt_flux}) $m^*$ is called ``the optimal flux''.

\subsection{Discretization}

We set $\Omega = [0,1]^d$. Let $\Omega^h$ be a grid on $\Omega$ with step size $h>0$: 
\[\Omega^h = \{0, h, 2h, 3h, \cdots, 1\}^d.\]  
Let $N = 1/h$ be the grid size. Any $x \in \Omega^h$ is a $d$ dimensional tensor, of which the value of the $i^{\text{th}}$ component $x_i$ is chosen from $\{0, h, 2h, 3h, \cdots, 1\}$. 
The discretized distributions $\rho^0_h,\rho^1_h$ are $(N+1)^d$ tensors, and the discretized flux $m_h$ is a $(N+1)^d \times d$ tensor, which represents a map $\Omega^h \to \Re^d$: $\rho^0_h = \{\rho^0(x)\}_{x \in \Omega^h}$, $\rho^1_h = \{\rho^1(x)\}_{x \in \Omega^h}$, and $\quad m_h = \{m(x)\}_{x \in \Omega^h}$. The discretized version of (\ref{eq:opt_flux}) can be written as
\begin{equation}
\label{eq:opt_flux_dis}
\begin{aligned}
\minimize_{m_h: \Omega^h \to \Re^d}& \sum_{x\in\Omega^h} ~~\|m_h(x)\|_p h^d\\
\text{subject to  }& ~~\text{div}^h (m_h(x)) = \rho^0_h(x) - \rho^1_h(x),\quad \forall  x \in \Omega^h,
\end{aligned}
\end{equation}
where the discrete divergence operator is: \[
\begin{aligned}
\text{div}^h (m_h(x)) =& \sum_{i=1}^d D_{h,i}m(x),\\ 
D_{h,i}m(x) =& \begin{cases}
(m_{h,i}(x_{-i},x_i))/h, & x_i = 0\\
(m_{h,i}(x_{-i},x_i) - m_{h,i}(x_{-i},x_i-h))/h, & 0<x_i < 1\\
(-m_{h,i}(x_{-i},x_i-h))/h, & x_i > 1.
\end{cases}
\end{aligned}
\]
In the definition of $\text{div}^h$, $m_{h}(x)\in\Re^d$ means the flow at point $x$, $m_{h,i}(x)\in\Re$ is the $i^{\text{th}}$ component of $m_{h}(x)$. The notion ``$-i$'' refers to all the components excluding $i$: $x_{-i} = \{x_j:j \in \{1,2,\cdots,d\},j\neq i\}$.

To simplify our notation, we rewrite the above problem (\ref{eq:opt_flux_dis}) as:
\begin{equation}
\label{eq:opt_flux_dis2}
\begin{aligned}
\minimize_{m_h: \Omega^h \to \Re^d}& ~~f(m_h) \\
\text{subject to }& ~~A_h m_h = \rho_h,
\end{aligned}
\end{equation}
where $f(\cdot)$ denotes a norm of $m_h$, $A_h$ denotes the divergence operator, which is linear, and $\rho_h=\rho_h^0-\rho_h^1$. 

The dual problem of (\ref{eq:opt_flux_dis2}), which is also the discrete version of (\ref{eq:opt_map_dual2}), is:
\begin{equation}
\label{eq:opt_dual3}
\begin{aligned}
\minimize_{\phi_h: \Omega^h \to \Re}& ~~\sum_{x\in\Omega^h}\phi_h(x)\rho_h(x) h^d \\
\text{subject to }& ~~\| A_h^* \phi_h(x)\|_q \leq 1,~\forall x\in\Omega^h,
\end{aligned}
\end{equation}
where $\phi_h:\Omega^h\to\Re$ is the \emph{Kantorovich potential}: $\phi_h = \{\phi(x)\}_{x\in\Omega^h}$.
The adjoint operator of $A_h$, $A^*_h$, denotes the gradient operator.

In this paper, we solve (\ref{eq:opt_flux_dis2}) and (\ref{eq:opt_dual3}) jointly by primal-dual algorithms.

\label{sec:stateoftheart} Define some norms on $\Omega^h$:
\[
\begin{aligned}
\|m_h\|_{2}^2 =& \sum_{x\in\Omega^h}\|m(x)\|^2_2, \quad \|m_h\|_{L^2}^2 = \sum_{x\in\Omega^h}\|m(x)\|^2_2 h^d,\\
\|\phi_h\|_{2}^2 =& \sum_{x\in\Omega^h}\phi^2(x), \quad
\|\phi_h\|_{L^2}^2 = \sum_{x\in\Omega^h}\phi^2(x) h^d,\\
\|\varphi_h\|_{2}^2 =&  \sum_{x\in\Omega^h}\|\varphi(x)\|^2_2, \quad \|\varphi_h\|_{L^2}^2 = \sum_{x\in\Omega^h}\|\varphi(x)\|^2_2 h^d.
\end{aligned}
\]

Define inner products on $\Omega^h$:
\[\langle \phi_h, \phi'_h \rangle = \sum_{x \in \Omega^h}\phi_h(x) \phi'_h(x),\]
\[\langle \phi_h, \phi'_h \rangle_h = \sum_{x \in \Omega^h}\phi_h(x) \phi'_h(x) h^d.\]

\section{Algorithm description}\label{section:3}
In this section, we review two recent primal-dual algorithms designed for (\ref{eq:opt_flux_dis2}) and (\ref{eq:opt_dual3}). We apply a multilevel framework (Section \ref{sec:multigrid_algo}) to further accelerate these algorithms. 
 
\subsection{Two recent algorithms for (\ref{eq:opt_flux_dis2}) and (\ref{eq:opt_dual3})}

\paragraph{Algorithm \ref{algo:primaldual} (Li et al.~\cite{li2018parallel})} Problems (\ref{eq:opt_flux_dis2}) and (\ref{eq:opt_dual3}) can be jointly solved by the following min-max problem: 
\begin{equation}
    \label{eq:minmax_cp}
    \min_{m_h}\max_{\phi_h} L(m_h,\phi_h), \quad \text{where } L(m_h,\phi_h) = f(m_h) + \langle \phi_h, A_h m_h - \rho_h \rangle_h.
\end{equation}

Inspired by the Chambolle-Pock Algorithm~\cite{chambolle2011first}, the authors of \cite{li2018parallel} proposed the following algorithm to solve (\ref{eq:minmax_cp}):
\begin{equation}
\label{eq:primaldual}
\begin{aligned}
m^{k+1}_h =& ~\argmin_{m_h} L(m_h,\phi^k_h) + \frac{1}{2\mu}\|m_h - m_h^k\|^2_{L^2},\\
\bar{m}^{k+1}_h =& ~2m^{k+1}_h - m^k_h,\\
\phi^{k+1}_h = & ~\argmax_{\phi_h} L(\bar{m}^{k+1}_h,\phi_h) - \frac{1}{2\tau}\|\phi_h - \phi_h^k\|^2_{L^2}.
\end{aligned}
\end{equation}
Parameters $\mu,\tau>0$ need to be tuned. If $\mu\tau \|A_h\|^2 < 1$, then we have the convergence $(m^k_h,\phi_h^k)\to(m^*_h,\phi_h^*)$, where $(m^*_h,\phi_h^*)$ is one of the solutions of (\ref{eq:minmax_cp}). In this paper, we use\footnote{The parameter choice $\mu = \tau = 1/(2\|A_h\|)$ is convenient for complexity analysis. Practically, $\mu = \tau = 1/\|A_h\|$ is better although it does not guarantee convergence theoretically.} $\mu = \tau = 1/(2\|A_h\|)$.
The iteration stops when the following fixed point residual (FPR) $\text{R}^k$ falls below a threshold:
\begin{equation}
\label{eq:pd_stop}
\text{R}^k_h := \frac{1}{\mu}\|m^{k+1}_h-m^k_h\|^2_{L^2} + \frac{1}{\tau}\|\phi^{k+1}_h-\phi^k_h\|^2_{L^2} - 2 \big\langle \phi^{k+1}_h-\phi^k_h, 
A_h(m^{k+1}_h-m^k_h) \big\rangle_h.
\end{equation}
The algorithm is summarized in Algorithm \ref{algo:primaldual}.

\begin{algorithm2e}
\SetKwInOut{input}{Input}
\input{Distributions $\rho^0, \rho^1$, grid step size $h$, initial 
point $m^0$, $\phi^0$,
tolerance $\varepsilon$.
}
{Start with $k=0$.}\\
\While{$R^k_h < \varepsilon$ is not satisfied} {
Execute (\ref{eq:primaldual}) and let $k \xleftarrow{} k+1$.
}
\KwOut{$m^{K}, \phi^{K}$. ($K$ is the number of iterations.)}
\caption{A primal-dual algorithm for EMD~\cite{li2018parallel}}\label{algo:primaldual}
\end{algorithm2e}

\paragraph{Algorithm \ref{algo:pdhg} (Jacobs et al.~\cite{jacobs2018solving})} Problem (\ref{eq:opt_flux_dis2}) can be written as:
\begin{equation}
    \label{eq:minmax_pdhg_original}
    \min_{m_h,u_h}\max_{\varphi_h} f(u_h) + \delta_{A_h m_h = \rho_h}(m_h) + \langle \varphi_h, m_h-u_h\rangle_h.
\end{equation}
where $\varphi_h:\Omega^h\to\Re^d$ is the dual variable and is also the gradient of the Kantorovich potential: $\varphi_h = A^*_h \phi_h$. Function $\delta_{A_h m_h = \rho_h}$ is the indicator function of $A_h m_h = \rho_h$: 
\[\delta_{A_h m_h = \rho_h}(m_h) = \begin{cases}
0, & \mathrm{if~} A_h m_h = \rho_h,\\
+\infty, & \mathrm{if~} A_h m_h \neq \rho_h.
\end{cases}\]
Define the convex conjugate of $f$:
\[f^*(\varphi_h) = \sup_{u_h} \langle \varphi_h, u_h \rangle_h - f(u_h).\]
Then (\ref{eq:minmax_pdhg_original}) is equivalent to the following problem: 
\begin{equation}
    \label{eq:minmax_pdhg}
    \min_{m_h}\max_{\varphi_h} \tilde{L}(m_h,\varphi_h), ~~ \text{where }\tilde{L}(m_h,\varphi_h)  = \delta_{A_h m_h = \rho_h}(m_h) -f^*(\varphi_h) + \langle \varphi_h, m_h \rangle_h.
\end{equation}

The authors of \cite{jacobs2018solving} solve (\ref{eq:minmax_pdhg}) in the following way:
\begin{equation}
\label{eq:pdhg}
\begin{aligned}
m^{k+1}_h =& ~\argmin_{m_h} \tilde{L}(m_h,\bar{\varphi}^k_h) + \frac{1}{2\mu}\|m_h - m^k_h\|^2_{L^2} \\
\varphi^{k+1}_h = & ~\argmax_{\varphi_h} \tilde{L}(m^{k+1}_h,\varphi_h) - \frac{1}{2\tau} \|\varphi_h - \varphi^k_h\|^2_{L^2}\\
\bar{\varphi}^{k+1}_h = & ~2 \varphi^{k+1}_h - \varphi^k_h,
\end{aligned}
\end{equation}
where the first subproblem solving $m^{k+1}_h$ requires computing a projection onto the affine space $\{m_h|A_h m_h = \rho_h\}$. Since the discrete Laplacian inverse  $((A_h)^*A_h)^{-1}$ can be easily computed by FFT, the projection could be efficiently calculated~\cite{jacobs2018solving}.

Parameters $\mu,\tau>0$ need to be tuned. As long as $\mu\tau < 1$, we have the convergence $(m^k_h,\varphi_h^k)\to(m^*_h,\varphi_h^*)$, where $(m^*_h,\varphi_h^*)$ is one of the solutions of (\ref{eq:minmax_pdhg}).
In this paper, we choose\footnote{The parameter choice $\mu = \tau = 1/2$ is convenient for complexity analysis. Practically, $\mu = \tau = 1$ is better.} $\mu = \tau = 1/2$.
The stopping condition is to have the following fixed point residual $\text{G}^k$ small enough:
\begin{equation}
\label{eq:pdhg_stop}
\text{G}^k_h = \frac{1}{\mu}\|m^{k+1}_h-m^k_h\|^2_{L^2} + \frac{1}{\tau}\|\varphi^{k+1}_h-\varphi^k_h\|^2_{L^2} + 2 \big\langle \varphi^{k+1}_h-\varphi^k_h, m^{k+1}_h-m^k_h \big\rangle_h.
\end{equation}
With $\varphi_h^*$ in hand, the Kantorovich potential $\phi^*_h$ can be easily found by the method given in the Appendix \ref{app:potential}. 
The algorithm is listed in Algorithm \ref{algo:pdhg}.

\begin{algorithm2e}
\SetKwInOut{input}{Input}
\input{Distributions $\rho^0, \rho^1$, grid step size $h$, initial point $m^0$, $\varphi^0$, tolerance $\varepsilon$.
}
Start with $k=0$.\\
\While{$G^k_h < \varepsilon$ is not satisfied} {
Execute (\ref{eq:pdhg}) and let $k \xleftarrow{} k+1$
}
\KwOut{$m^{K},\varphi^K$. ($K$ is the number of iterations.)}
\caption{Prox-PDHG for EMD~\cite{jacobs2018solving}}\label{algo:pdhg}
\end{algorithm2e}

\subsection{A framework: multilevel initialization}
\label{sec:multigrid_algo}

In this subsection, we 
describe a framework, inspired by the cascadic multilevel method \cite{bornemann1996cascadic}, to substantially speed up Algorithms \ref{algo:primaldual} and \ref{algo:pdhg}. With the multilevel framework,  Algorithms \ref{algo:primaldual} and \ref{algo:pdhg} lead to Algorithms \ref{algo:multigrid1} and \ref{algo:multigrid2} respectively.

Suppose we have $L$ levels of grids with step sizes $h_1, h_2, \cdots, h_L$ respectively. The step sizes satisfy \[h_1 > h_2 > \cdots > h_{L-1} > h_L = h.\]The finest step size $h_L=h$. On each level, the space $\Omega$ is respectively discretized as \[\Omega^{h_1},\cdots,\Omega^{h_{L-1}},\Omega^{h_L}.\]
If we take $h_l = 2^{L-l}h$, then we have 
\[
\Omega^{h_1} \subset \Omega^{h_2} \subset \cdots \Omega^{h_{L-1}} \subset \Omega^{h}. 
\]
On the $l^{\text{th}}$ level, the optimal flux problem (\ref{eq:opt_flux_dis2}) is 
\begin{equation}
\label{eq:opt_flux_l}
\begin{aligned}
\minimize_{m_{h_l}: \Omega^{h_l} \to \Re^d}& f(m_{h_l}) \\
\text{subject to }& A_{h_l} m_{h_l} = \rho_{h_l}.
\end{aligned}
\end{equation}

We apply the cascadic multilevel technique \cite{bornemann1996cascadic} to the OT problem. We use $0$ initial solution on the level $l=1$ and solve a sequence of minimization problem (\ref{eq:opt_flux_l}) with one pass from the coarsest level $l=1$ to the finest level $l=L$. 
On each level, we use Algorithm \ref{algo:primaldual} or Algorithm \ref{algo:pdhg} that is stopped as the iterate is accurate enough ($R^k_{h_1}<\varepsilon_l$ for Algorithm \ref{algo:primaldual}, $G^k_{h_1}<\varepsilon_l$ for Algorithm \ref{algo:pdhg}). The obtained solution is denoted by $(m^K_{h_1},\phi^K_{h_1})$ or $(m^K_{h_1},\varphi^K_{h_1})$. After that, we interpolate the obtained solutions to the next level $l=2$ and treat them as the initial solutions of level $l=2$. 
The process is repeated for $l=3,\cdots,L$.
Algorithms \ref{algo:multigrid1} and \ref{algo:multigrid2} are the multilevel versions of Algorithms \ref{algo:primaldual} and \ref{algo:pdhg} respectively.

\begin{protocol}
    \SetKwInOut{input}{Input}
    \SetKwInOut{initial}{Initialization}
    \input{$\rho^0, \rho^1$, $h$, $L$, a sequence of tolerances $\{\varepsilon_l\}_{l=1}^L$.}
    \initial{Let $m^{K_0}_{h_0}=0, \phi^{K_0}_{h_0}=0.$}
    \For{$l=1,2,\cdots, L$} {
    Initialize the current level:
    \[m^0_{h_l} = \text{Interpolate
    }(m^{K_{l-1}}_{h_{l-1}}), \quad
    \phi^0_{h_l} = \text{Interpolate }(\phi^{K_{l-1}}_{h_{l-1}})
    \]\\
    Call Algorithm \ref{algo:primaldual}: \[(m^{K_{l}}_{h_l}, \phi^{K_{l}}_{h_l}) = \text{Algorithm \ref{algo:primaldual}}(\rho^0, \rho^1, h_l, m^0_{h_l},\phi^0_{h_l},\varepsilon_l)\]
    }
    \KwOut{$m^{K_L}_{h_L},\phi^{K_L}_{h_L}$
    }
    \caption{Multilevel version of Algorithm \ref{algo:primaldual}}
    \label{algo:multigrid1}
    \end{protocol}

     \begin{protocol}
    \SetKwInOut{input}{Input}
    \SetKwInOut{initial}{Initialization}
    \input{$\rho^0, \rho^1$, $h$, $L$, a sequence of tolerances $\{\varepsilon_l\}_{l=1}^L$.}
    \initial{Let $m^{K_0}_{h_0}=0, \varphi^{K_0}_{h_0}=0$.}
    \For{$l=1,2,\cdots, L$} {
    Initialize the current level: 
    \[m^0_{h_l} = \text{Interpolate }(m^{K_{l-1}}_{h_{l-1}}),\quad 
    \varphi^0_{h_l} = \text{Interpolate }(\varphi^{{K_{l-1}}}_{h_{l-1}})
    \]\\
    Call Algorithm \ref{algo:pdhg}: \[(m^{K_l}_{h_l}, \varphi^{K_l}_{h_l}) = \text{Algorithm \ref{algo:pdhg}}(\rho^0, \rho^1, h_l, m^0_{h_l},\varphi^0_{h_l},\varepsilon_l)\]
    }
    \KwOut{$m^{K_L}_{h_L},\phi^{K_L}_{h_L}$ (Obtain $\phi^{K_L}_{h_L}$ from $\varphi^{K_L}_{h_L}$, see Appendix \ref{app:potential}) 
    }
    \caption{Multilevel version of Algorithm \ref{algo:pdhg}}
    \label{algo:multigrid2}
    \end{protocol}

\subsection{Cross-level interpolation}
\label{sec:interpolate}

In this subsection, we describe the cross-level interpolations in Algorithms \ref{algo:multigrid1} and \ref{algo:multigrid2} in detail.

\paragraph{Interpolation of potentials $\phi_h$} For any $x\in \Omega^{h_{l}}$ on level $l$, we partition the set of the coordinates $x_j$ into two subsets, depending on whether they belong to the grid on the coarser level $l-1$: 
\begin{equation}
    \label{eq:J}
    \begin{aligned}
    J_l =& \big\{j: x_{j} \in \{0, h_{l-1}, 2h_{l-1}, \cdots, 1\}\big\},\\ 
    \hat{J}_l =& \big\{j: x_{j} \in \{0, h_{l}, 2h_{l}, \cdots, 1\}, x_{j} \notin \{0, h_{l-1}, 2h_{l-1}, \cdots, 1\}\big\}.
    \end{aligned}
\end{equation}
Define a partial neighborhood of $x$:
\[\mathcal{N}_l(x) = \Big\{y\in\Omega^{h_{l-1}}: y_j = x_j,~~ \forall j \in J_l; ~~~ |y_j - x_j| \leq h_l,~~\forall j \in \hat{J}_l.\Big\}.\]
The mapping $\phi_{h_l} =$ Interpolate $(\phi_{h_{l-1}})$ is defined pointwise as:
\begin{equation}
    \label{eq:inter_phi}
    \phi_{h_l}(x) =  \frac{1}{|\mathcal{N}_l(x)|} \sum_{y \in \mathcal{N}_l(x)} \phi_{h_{l-1}}(y),\quad \forall x \in \Omega^{h_l}.
\end{equation}
For example, if $d=2$ (2D case) and $h_l = 2^{-(l-1)},l\geq 2$, (\ref{eq:inter_phi}) can be written as:
\[
\phi_{h_l}(x_1,x_2) = 
\begin{cases}
\phi_{h_{l-1}}(x_1,x_2),\qquad \text{if }(x_1,x_2) \in \mathcal{X}_1,&\\
\Big(\phi_{h_{l-1}}(x_1,x_2-h_l)+\phi_{h_{l-1}}(x_1,x_2+h_l)\Big)/2,\qquad  \text{if }(x_1,x_2) \in \mathcal{X}_2,&\\
\Big(\phi_{h_{l-1}}(x_1-h_l,x_2)+\phi_{h_{l-1}}(x_1+h_l,x_2)\Big)/2,\qquad \text{if }(x_1,x_2) \in \mathcal{X}_3,&\\
\Big(\phi_{h_{l-1}}(x_1-h_l,x_2-h_l) + \phi_{h_{l-1}}(x_1-h_l,x_2+h_l)...\\ + \phi_{h_{l-1}}(x_1+h_l,x_2-h_l) + \phi_{h_{l-1}}(x_1+h_l,x_2+h_l)\Big)/4,~~ \text{otherwise},&\\
\end{cases}
\]
where $\mathcal{X}_1 = \{(x_1,x_2): x_1|h_{l-1},x_2|h_{l-1}\}$, $\mathcal{X}_2 = \{(x_1,x_2): x_1|h_{l-1},x_2\not |h_{l-1}\}$ and $\mathcal{X}_3 = \{(x_1,x_2): x_1\not |h_{l-1},x_2|h_{l-1}\}$. 
Figure \ref{fig:iterphi} gives an illustration of this 2D interpolation.

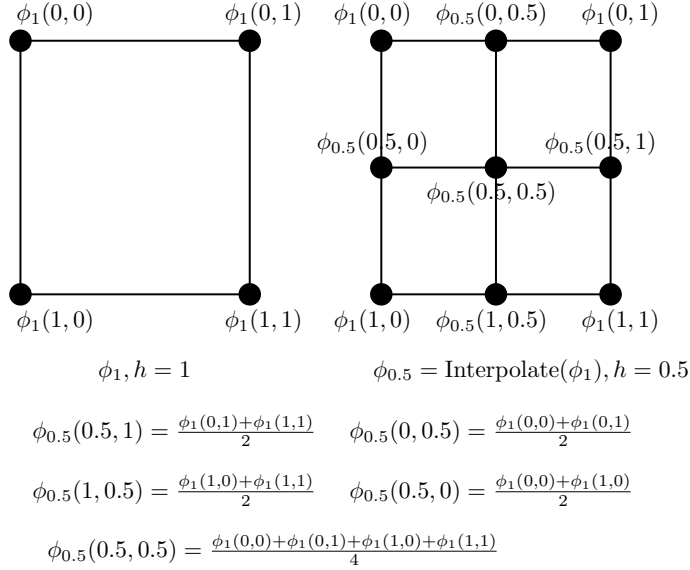
\begin{figure}[t]
    \centering
    \input{iter_phi.tikz}
    \caption{An illustration of (\ref{eq:inter_phi}) (2D case): from $1\times 1$ grid to $2\times 2$ grid}
    \label{fig:iterphi}
\end{figure}

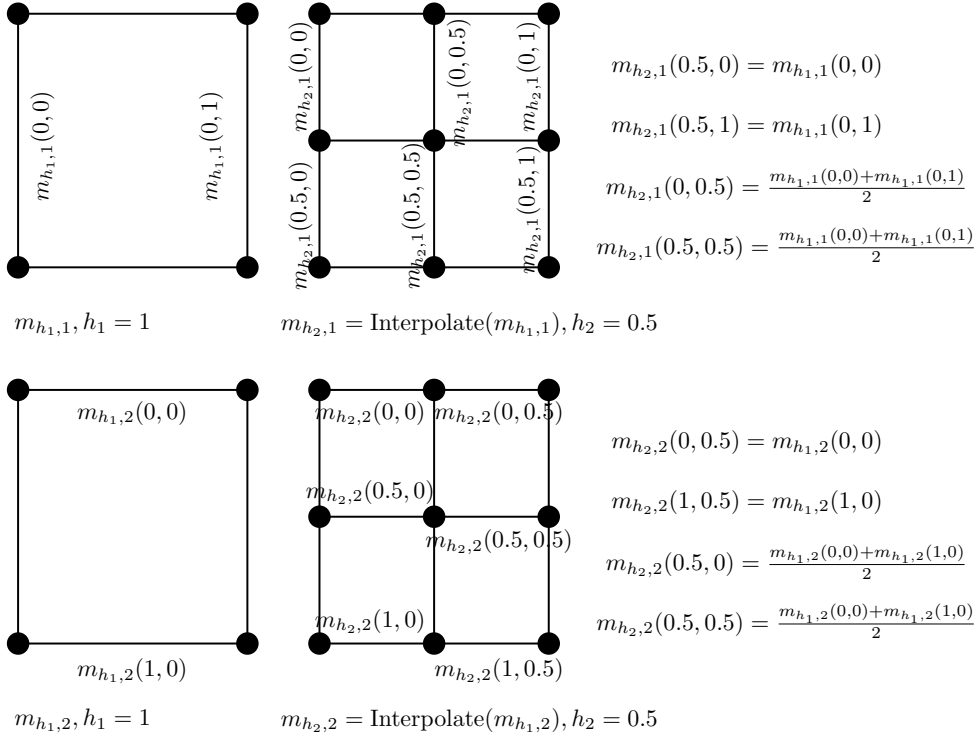
\begin{figure}[t]
    \centering
    \input{iter_m.tikz}
    \input{iter_m2.tikz}
    \caption{An illustration of (\ref{eq:inter_m}) (2D case): from $1\times 1$ grid to $2\times 2$ grid}
    \label{fig:iterm}
\end{figure}

\paragraph{Interpolation of flux $m_h$} Due to the zero-flux boundary condition for (\ref{eq:opt_flux}),
interpolating $m$ is different from $\phi$. The flow $m$ can be viewed as ``edge weights'' on the grid~\cite{li2018parallel,bassetti2018computation}, as in Figure \ref{fig:iterm}. 
For the $i$-th coordinate $x \in \Omega^{h_l}_i$, we use nearest neighbor interpolation: $y_{\mathrm{NN}} \in \argmin_{\alpha \in\Omega^{h_{l-1}}_i}|\alpha -x_i|$; for the other coordinates, we use the same method with (\ref{eq:inter_phi}). 
Define a neighborhood related with direction $i$:
\[\hat{\mathcal{N}}_{l,i}(x) = \Big\{ y\in\Omega^{h_{l-1}}: y_j = x_j, \forall j \in J_l; ~~y_i = y_{\mathrm{NN}};~~ |y_j - x_j| \leq h_l,\forall j \in \hat{J}_l/\{i\}. \Big\}\]
Then the mapping $m_{h_l} = \text{Interpolate }(m_{h_{l-1}})$ is pointwise defined as follows:
\begin{equation}
    \label{eq:inter_m}
    \begin{aligned}
m_{h_l,i}(x) = 
\begin{cases}
\frac{1}{|\mathcal{N}_l(x)|}  \sum_{y \in \mathcal{N}_l(x)} m_{h_{l-1},i}(y),\quad & i \in J_l, \\
\frac{1}{|\hat{\mathcal{N}}_{l,i}(x) |}  \sum_{y \in \hat{\mathcal{N}}_{l,i}(x) } m_{h_{l-1},i}(y),\quad &
i \in \hat{J}_l.
\end{cases}
\end{aligned}
\end{equation}
For example, if $d=2$ (2D case) and $h_l = 2^{-(l-1)},l\geq2$, (\ref{eq:inter_m}) can be written as:
\[
m_{h_l,1}(x_1,x_2) = 
\begin{cases}
m_{h_{l-1},1}(x_1,x_2),\qquad\qquad\qquad \qquad\qquad\qquad\qquad\text{if }(x_1,x_2) \in \mathcal{X}_1,\\
m_{h_{l-1},1}(x_1-h_l,x_2),\qquad \qquad\qquad\qquad\qquad\qquad\text{if }(x_1,x_2) \in \mathcal{X}_3,\\
\Big(m_{h_{l-1},1}(x_1,x_2-h_l)+m_{h_{l-1},1}(x_1,x_2+h_l)\Big)/2,\text{if }(x_1,x_2) \in \mathcal{X}_2,\\
\Big(m_{h_{l-1},1}(x_1-h_l,x_2-h_l)+m_{h_{l-1},1}(x_1-h_l, x_2+h_l)\Big)/2,\\ \qquad\qquad\qquad\qquad\qquad\qquad\qquad\qquad \qquad\qquad\qquad~~\text{otherwise},& \\
\end{cases}
\]
\[
m_{h_l,2}(x_1,x_2) = 
\begin{cases}
m_{h_{l-1},2}(x_1,x_2),\qquad\qquad\qquad \qquad\qquad\qquad\qquad\text{if }(x_1,x_2) \in \mathcal{X}_1,\\
m_{h_{l-1},2}(x_1,x_2-h_l),\qquad \qquad\qquad\qquad\qquad\qquad\text{if }(x_1,x_2) \in \mathcal{X}_2,\\
\Big(m_{h_{l-1},2}(x_1-h_l,x_2)+m_{h_{l-1},2}(x_1+h_l,x_2)\Big)/2,\text{if }(x_1,x_2) \in \mathcal{X}_3,\\
\Big(m_{h_{l-1},2}(x_1-h_l,x_2-h_l)+m_{h_{l-1},2}(x_1+h_l,x_2-h_l)\Big)/2,\\ \qquad\qquad\qquad\qquad\qquad\qquad\qquad\qquad\qquad\qquad \qquad\text{otherwise}.& \\
\end{cases}
\]
Figure \ref{fig:iterm} illustrates the above formula.

\paragraph{Interpolation of $\varphi_h$} Since $\varphi_h$ has the same dimension as $m_h$, the interpolation of $\varphi_h$ is the same as interpolation of flux $m_h$ (\ref{eq:inter_m}).

The interpolation operators introduced above are linear operators satisfying the following properties that are proved in Appendix \ref{app:inter}:
   
\begin{lemma}
\label{lemma:inter_bdd}
If $h_{l-1} = 2 h_{l}$, then we have 
\begin{equation}
    \label{eq:interpolate_bdd}
    \begin{aligned}
    \| \mathrm{Interpolate }(\phi_{h_{l-1}}) \|^2_{L^2} \leq& \| \phi_{h_{l-1}} \|^2_{L^2}, \quad \forall \phi_{h_{l-1}}: \Omega^{h_{l-1}}\to\Re. \\
    \| \mathrm{Interpolate }(m_{h_{l-1}}) \|^2_{L^2} \leq& \| m_{h_{l-1}} \|^2_{L^2}, \quad \forall m_{h_{l-1}}: \Omega^{h_{l-1}}\to\Re^d.\\
    \| \mathrm{Interpolate }(\varphi_{h_{l-1}}) \|^2_{L^2} \leq& \| \varphi_{h_{l-1}} \|^2_{L^2}, \quad \forall \varphi_{h_{l-1}}: \Omega^{h_{l-1}}\to\Re^d.
    \end{aligned}
\end{equation}
\end{lemma}

\subsection{Parameter choices}
\label{sec:prac-tol}
The choice of stopping tolerances ($\varepsilon$ for Algorithms \ref{algo:primaldual} and \ref{algo:pdhg}; $\{\varepsilon_l\}_{l=1}^L$ for Algorithms \ref{algo:multigrid1} and \ref{algo:multigrid2}) are critical to the practical performances of the algorithms. We will discuss the theoretical conditions for the tolerances in Section \ref{section:4} and empirical formulas in Section \ref{section:5}.  To summarize, we list typical rules in 2D ($d=2$) case below:
\begin{itemize}
    \item Algorithm \ref{algo:primaldual}: $\varepsilon = O(h^3)$.
    \item Algorithm \ref{algo:pdhg}: $\varepsilon = O(h^2)$.
    \item Algorithm \ref{algo:multigrid1}: $\varepsilon_l = \varepsilon_L\times(h_l/h_L)^{-1}$, where $\varepsilon_L = O(h^2)$.
    \item Algorithm \ref{algo:multigrid2}: $\varepsilon_l = \varepsilon_L\times(h_l/h_L)^{-1}$, where $\varepsilon_L = O(h^2)$.
\end{itemize}
With these rules, the algorithms provide outputs with errors in the order of the discretization error. 
Moreover, we take $L=O(\log(1/h))$.

\section{Analysis of computational costs}\label{section:4}
In this section, we provide a complexity analysis of Algorithms \ref{algo:primaldual}, \ref{algo:pdhg}, \ref{algo:multigrid1} and \ref{algo:multigrid2}. We introduce compact notions $z_h = (m_h, \phi_h)$ for Algorithms \ref{algo:primaldual},\ref{algo:multigrid1} and $y_h = (m_h, \varphi_h)$ for Algorithms \ref{algo:pdhg},\ref{algo:multigrid2}.  Let $Z^*_{h},Y^*_{h}$, respectively, be the solution set of the  following min-max problems with grid step size $h$:
\[
\begin{aligned}
Z^*_{h} =& \Big\{(m_{h}^*,\phi_{h}^*)\Big|(m_{h}^*,\phi_{h}^*) \mathrm{~is~a~saddle~point~of~} L(m_{h},\phi_{h}) \Big\},\\
Y^*_{h} =& \Big\{(m_{h}^*,\varphi_{h}^*)\Big|(m_{h}^*,\varphi_{h}^*) \mathrm{~is~a~saddle~point~of~} \tilde{L}(m_{h},\varphi_{h}) \Big\},
\end{aligned}
\]
where $L$ and $\tilde{L}$ are defined in (\ref{eq:minmax_cp}) and (\ref{eq:minmax_pdhg}) respectively.

\subsection{Assumptions} In this subsection, some assumptions used in our theories will be introduced and discussed.

\begin{assume}
\label{assume:cp2}
The solution sets $Z^*_{h_l}$ on all the levels are nonempty. And there exists a constant $C_1$ independent of $h_l$ such that,
\begin{equation}
\label{eq:l2bdd_z}
    \|z^*_{h_l}\|^2_{L^2} \leq C_1,\quad \forall z^*_{h_l}\in Z^*_{h_l}, \quad l = 1,2,\cdots,L.
\end{equation}
\end{assume}

Assumption \ref{assume:cp2} is mild. Since $z^*_{h_l} = (m^*_{h_l},\phi^*_{h_l})$, the norm of $z^*_{h_l}$ can be decomposed as $ \|z^*_{h_l}\|^2_{L^2} =  \|m^*_{h_l}\|^2_{L^2} +  \|\phi^*_{h_l}\|^2_{L^2}$. 
 The dual solution $\phi^*_{h_l}$, by the definition in (\ref{eq:opt_dual3}), has the property: $\|A_h^* \phi^*_{h_l}(x)\|_q\leq 1,\forall x \in \Omega^{h_l}$, where $A_h^*$ is the gradient operator defined on $\Omega^{h_l}$. It implies that all the dual solutions $\phi^*_{h_l}$ are Lipschitz continuous uniformly on the compact domain $\Omega=[0,1]^d$. Thus, all the dual solutions $\phi^*_{h_l}$ are uniformly bounded as long as they are kept with a zero mean. Actually keeping $\phi^*_{h_l}$ to be zero-meaned is not difficult, see \cite{villani2008optimal}.
 The primal solution $m^*_{h_l}$, by definition, is the solution of minimization problem (\ref{eq:opt_flux_l}). 
 It suffices to have $L2$-boundedness of $m^*_{h_l}$ if $\rho^0,\rho^1$ are smooth. In that case, $\rho^0_{h_l},\rho^1_{h_l}$ are similar on different levels and so does $m^*_{h_l}$, then (\ref{eq:l2bdd_z}) can be expected.
 Generally speaking, on some extreme cases like $\rho^0$ and $\rho^1$ are  $\delta$-functions, it does not hold. 
 However, on commonly used examples, we numerically validated Assumption \ref{assume:cp2} in Table \ref{tab:assume2} and observed that $C_1$ exists and is independent of grid size. 

\begin{assume}
\label{assume:cp}
For any 
optimal solution $z^*_{h_l}\in Z^*_{h_l}$ on level $l$, there exists an 
optimal solution $z^*_{h_{l+1}}\in Z^*_{h_{l+1}}$ on the finer level $l+1$ such that
\begin{equation}
\label{eq:assumecp}
    \|\mathrm{Interpolate }(z^*_{h_l}) - z^*_{h_{l+1}}\|^2_{L^2} \leq C_2 (h_l)^{r},\quad  l = 1,2,\cdots,L-1,
\end{equation}
where $C_2>0$ is independent of $h_l$ and $r > 0$ depends on the smoothness of the solution $z^*_{h_l}$, the interpolation method we choose and the properties of $\rho^0_{h_l}$ and $\rho^1_{h_l}$ on each of the levels.
\end{assume}

Assumption \ref{assume:cp} requires the solution sets between two consecutive levels are close to each other and each solution is smooth enough.
Ideally speaking, if $z^*_{h_l}$ and $z^*_{h_{l+1}}$ are the discretized version of the same underlying continuous $z^*$ and $z^*$ is smooth enough, it should hold $r=2$ due to our interpolation method.
Generally speaking, one can only expect $z^*_{h}\to z^*$ as $h\to 0$ without knowing the convergence rate~\cite{benamou2018numerical} and the smoothness of $z^*$ cannot be guaranteed (for example, $\rho^0$ and $\rho^1$ are taken as $\delta$-functions).
Thus, we are not able to show (\ref{eq:assumecp}) holds  theoretically.
However, Assumption \ref{assume:cp} holds on commonly used examples. We numerically validated it in 
Section \ref{sec:assumes} with $d=2$ and $p=1,2,\infty$. 
Figures \ref{fig:assume1m} and \ref{fig:assume1phi} give a visualized example that has multiple solutions on different levels. Table \ref{tab:assume2} quantifies $\|\mathrm{Interpolate }(z^*_{h_l}) - z^*_{h_{l+1}}\|^2_{L^2}$ and shows that $r$ is approximately $r \approx 2$.

\begin{assume}
\label{assume:pdhg2}
The solution sets $Y^*_{h_l}$ on all the levels are nonempty. And there exists a constant $C_3$ independent of $h_l$ such that,
\begin{equation}
    \|y^*_{h_l}\|^2_{L^2} \leq C_3,\quad \forall y^*_{h_l}\in Y^*_{h_l},\quad  l = 1,2,\cdots,L.
\end{equation}
\end{assume}

Similar with the discussion following Assumption \ref{assume:cp2}, $ \|y^*_{h_l}\|^2_{L^2} =  \|m^*_{h_l}\|^2_{L^2} +  \|\varphi^*_{h_l}\|^2_{L^2}$. 
The dual optimal solution $\varphi^*_{h_l}$, by the definition in (\ref{eq:minmax_pdhg}), has the property: $ \varphi^*_{h_l} = A_h^* \phi^*_{h_l}$. Since $\|A_h^* \phi^*_{h_l}(x)\|_q\leq 1,\forall x \in \Omega^{h_l}$, we have $\|\varphi^*_{h_l}(x)\|_q\leq 1,\forall x \in \Omega^{h_l}$, which implies all the dual solutions $\phi^*_{h_l}$ are uniformly bounded:\footnote{This bound is due to the fact that $\|a\|_2 \leq \sqrt{d} \|a\|_q $ for all $a\in\Re^d$ and $1 \leq q \leq \infty$.}
\[\|\varphi^*_{h_l}\|^2_{L^2} = \sum_{x\in\Omega^{h_l}}\|\varphi^*_{h_l}(x)\|^2_2 h_l^d \leq \sum_{x\in\Omega^{h_l}} d \|\varphi^*_{h_l}(x)\|^2_q h_l^d \leq d.\]

\begin{assume}
\label{assume:pdhg}
For any 
optimal solution $y^*_{h_l}\in Y^*_{h_l}$ on level $l$, there exists an 
optimal solution $y^*_{h_{l+1}}\in Y^*_{h_{l+1}}$ on the finer level $l+1$ such that
\begin{equation}
    \|\mathrm{Interpolate }(y^*_{h_l}) - y^*_{h_{l+1}}\|^2_{L^2} \leq C_4 (h_{l})^{\nu},\quad l = 1,2,\cdots,L-1,
\end{equation}
where $C_4>0$ is independent of $h_l$ and $\nu > 0$ depends on the smoothness of the solution $y^*_{h_l}$, the interpolation method we choose and the properties of $\rho^0_{h_l}$ and $\rho^1_{h_l}$ on each of the levels.
\end{assume}

Table \ref{tab:assume2} 
shows that $\nu$ is approximately $\nu\approx 1$ on the commonly used examples. The factor $\nu$ is smaller than the factor $r$ in Assumption \ref{assume:cp} ($r\approx 2$ on the same examples) because $\varphi^*_{h_l}$ is the gradient of $\phi^*_{h_l}$: $ \varphi^*_{h_l} = A_h^* \phi^*_{h_l}$ and $\varphi^*_{h_l}$ is less smooth than $\phi^*_{h_l}$. Figures \ref{fig:assume1phi} and \ref{fig:assume3p} visualize $\phi^*_{h_l}$  and $\varphi^*_{h_l}$ respectively and  support this point.

\subsection{Analysis of Algorithms \ref{algo:primaldual} and \ref{algo:multigrid1}} 

Assumption (\ref{eq:assumecp}) tells us the solution on the $l$-th level $z^*_{h_l}$ has an grid error up to $O((h_l)^{r})$ compared with the solution on a finer level $z^*_{h_{l+1}}$. Thus, on each level, it's enough to do iteration until $\|z^{K_l}_{h_{l}}-z^*_{h_{l}}\|^2_{L^2} \leq O( (h_{l})^{r})$. The following lemma demonstrates that, as long as we choose $\varepsilon$ (for Algorithm \ref{algo:primaldual}) and $\{\varepsilon_l\}_{l=1}^L$ (for Algorithm \ref{algo:multigrid1}) properly, such condition can be satisfied.

\begin{lemma}
\label{lemma:epsilon}
Given $\rho^0,\rho^1,h$, we choose $h_l = 2^{L-l}h$. 
\begin{enumerate}
    \item There exists $\bar{\varepsilon}$, as long as we run Algorithm \ref{algo:primaldual} with any $\varepsilon \leq \bar{\varepsilon}$, we have
    \begin{equation}
    \label{eq:lemma11}
        \|z^K_{h}-z^*_{h}\|^2_{L^2} \leq C_2 h^{r},\quad \text{for some }z^*_{h} \in Z^*_{h}.
    \end{equation}
    \item Suppose we run Algorithm \ref{algo:multigrid1} with arbitrary tolerances $\{\varepsilon_i\}_{i=1}^{l-1}$ on the lowest $l-1$ levels. There exists a threshold $\bar{\varepsilon}_{l}$ depending on $\{\varepsilon_i\}_{i=1}^{l-1}$ , as long as we set $\varepsilon_{l} \leq \bar{\varepsilon}_{l}\Big(\{\varepsilon_i\}_{i=1}^{l-1}\Big)$ on level $l$, we have
    \begin{equation}
    \label{eq:lemma12}
        \|z^{K_l}_{h_{l}}-z^*_{h_{l}}\|^2_{L^2} \leq C_2 (h_{l})^{r},\quad \text{for some }z^*_{h_{l}} \in Z^*_{h_{l}}.
    \end{equation}
\end{enumerate}
\end{lemma}
\begin{proof}\textbf{Step 1:} 
We consider Algorithm \ref{algo:primaldual}. Define 
 \[M_{h} = {h}^d\begin{bmatrix}
 I/\mu & -(A_{h})^* \\ -A_{h} & I/\tau
 \end{bmatrix}.\]
  Then the fixed point residual can be written as $R^k_{h} = \|z^{k+1}_{h} -z^k_{h}\|^2_{M_{h}}$, and Chambolle-Pock is equivalent to the proximal point algorithm (PPA) with the $M_{h}$-metric (Theorem 1 in \cite{li2018parallel}). Given $\mu\tau\|A_{h}\|^2<1$, $R^k_h$ is monotone:
 \begin{equation}
     \label{eq:proof_cp_2}
     R^{k+1}_{h} \leq  R^k_{h}, \quad \forall k,
 \end{equation}
 and the global convergence holds in the sense:
  \begin{equation}
     \label{eq:proof_cp_3}
     z^k_{h} \to z^*_{h},\quad \text{for some }z^*_{h} \in Z^*_{h}, \text{ as }k \to \infty.
 \end{equation}
 The global convergence (\ref{eq:proof_cp_3}) means that there exists a  $\bar{K}$, 
 such that $\|z^k_{h}-z^*_{h}\|^2_{L^2} \leq C_2 h^{r}$ holds
for all $k \geq \bar{K}$. Let $\bar{\varepsilon} = R^{\bar{K}}_{h}$ and $K$ be the number of iterations when the stopping condition 
$R^{K}_{h} \leq \varepsilon$ is satisfied. 
 As long as $\varepsilon < \bar{\varepsilon}$, we have \[R^{K}_{h} \leq \varepsilon < \bar{\varepsilon} = R^{\bar{K}}_{h}.\]
Then, by the monotonicity of FPR (\ref{eq:proof_cp_2}), we conclude that $K \geq \bar{K}$. Inequality (\ref{eq:lemma11}) is proved.

\textbf{Step 2:} Next we consider Algorithm \ref{algo:multigrid1}. On the first level $l=1$, we repeat the same argument as above, and we obtain that, there exists a threshold $\bar{\varepsilon}_1$ such that (\ref{eq:lemma12}) holds for $l=1$ as long as $\varepsilon_1 \leq \bar{\varepsilon}_1$. Similarly, for $l \geq 2$, there exists a threshold $\bar{\varepsilon}_l$ such that (\ref{eq:lemma12}) holds for $l\geq 2$ as long as $\varepsilon_l \leq \bar{\varepsilon}_l$. Now let's analyze the dependencies of $\bar{\varepsilon}_l$ and $\{\varepsilon_i\}_{i=1}^{l-1}$.
Suppose $\{\varepsilon_i\}_{i=1}^{l-1}$ are given. By running Algorithm \ref{algo:multigrid1} on level $1,2,\cdots,l-1$, we are able to determine the solution of the $(l-1)$-th level $z^{K_{l-1}}_{h_{l-1}}$. Conducting interpolation on $z^{K_{l-1}}_{h_{l-1}}$, we obtain $z^0_{h_{l}}$. In another word, $z^0_{h_{l}}$ is determined by the threshold sequence $\{\varepsilon_i\}_{i=1}^{l-1}$.
Since $\bar{\varepsilon}_l$ depends on $z^0_{h_{l}}$ and $z^0_{h_{l}}$ depends on $\{\varepsilon_i\}_{i=1}^{l-1}$, we can write $\bar{\varepsilon}_l$ as a function of $\{\varepsilon_i\}_{i=1}^{l-1}$: $\bar{\varepsilon}_{l}\Big(\{\varepsilon_i\}_{i=1}^{l-1}\Big)$. For any $\varepsilon_{l} \leq \bar{\varepsilon}_{l}\Big(\{\varepsilon_i\}_{i=1}^{l-1}\Big)$,  (\ref{eq:lemma12}) holds.
Lemma \ref{lemma:epsilon} is proved.
\end{proof}

We define the set of ``good'' stopping tolerances for Algorithms \ref{algo:primaldual} and \ref{algo:multigrid1}:
\begin{align}
 \mathcal{T}_{\mathrm{1}} = & \bigg \{\varepsilon \bigg | \varepsilon \leq \bar{\varepsilon}\bigg \}\label{eq:epsilon1},\\ \mathcal{T}_{\mathrm{1M}} =  & \bigg  \{\{\varepsilon_l\}_{l=1}^L \bigg  | \varepsilon_1 \leq \bar{\varepsilon}_1, ~\varepsilon_2 \leq \bar{\varepsilon}_2(\varepsilon_1),~\cdots,~\varepsilon_L\leq \bar{\varepsilon}_{L}\Big(\{\varepsilon_l\}_{l=1}^{L-1}\Big)\bigg  \}\label{eq:epsilon1m}.
\end{align}
Based on Lemma \ref{lemma:epsilon}, $\mathcal{T}_{\mathrm{1}}$ and $\mathcal{T}_{\mathrm{1M}}$ are nonempty.
Given these stopping tolerances, the performance of Algorithms \ref{algo:primaldual} and \ref{algo:multigrid1} can be analyzed as follows.

\begin{theorem}
\label{lemma:cp}
 Given $\rho^0,\rho^1,h,L\geq 2$ and $h_l = 2^{L-l}h$ and Assumptions \ref{assume:cp2},\ref{assume:cp}, if we take $0$ as the initialization, then the following holds\footnote{In this article, $O(\cdot)$ denotes the asymptotic rate as $\varepsilon\to 0$ and $h\to 0$.} :
\begin{enumerate}
    \item With $\varepsilon \in \mathcal{T}_{\mathrm{1}}$, Algorithm \ref{algo:primaldual} takes $O(\frac{1}{\varepsilon}\frac{\sqrt{d}}{h})$ iterations to stop.
    \item With $\{\varepsilon_l\}_{l=1}^L \in \mathcal{T}_{\mathrm{1M}}$, Algorithm \ref{algo:multigrid1} takes 
    $O(\frac{1}{\varepsilon_l}\frac{\sqrt{d}}{h_l^{1-r}})$
    iterations on level $l$ ($l \geq 2$). 
\end{enumerate}
\end{theorem}

This theorem shows why Algorithm \ref{algo:multigrid1} helps speed up Algorithm \ref{algo:primaldual}. As long as the optimal solution on the coarse level is close to one of the optimal solutions on the finer level, the multilevel technique is able to reduce the number of iterations on the finer level.
If the distance between the coarse solution and  the fine solution is controlled by $O(h^{r})$ (Assumption \ref{assume:cp}), the number of iterations can be reduced by $h^{r}$. 
Specifically, with $\varepsilon = \varepsilon_L$, Algorithm \ref{algo:primaldual} takes $O(\frac{1}{\varepsilon}\frac{\sqrt{d}}{h})$ iterations while Algorithm \ref{algo:multigrid1} takes $O(\frac{1}{\varepsilon}\frac{\sqrt{d}}{h^{1-r}})$ on the finest level because $h=h_L$. Although Algorithm \ref{algo:multigrid1} leads to extra calculations on coarser grids, the advantage of our multilevel algorithm is able to overcome the extra costs. 
Table \ref{tab:pd_comparel} validates this point. Moreover, the choice of $\varepsilon = \varepsilon_L$ is reasonable, even over-fair, in practice. Table \ref{tab:1m} shows that Algorithm \ref{algo:multigrid1} can obtain better solution than Algorithm \ref{algo:primaldual} with $\varepsilon_L > \varepsilon$.

\begin{proof} \textbf{Step 1:} Analyzing how many iterations Algorithm \ref{algo:primaldual} takes. 
  As we discussed in the proof of Lemma \ref{lemma:epsilon}, Algorithm \ref{algo:primaldual} is equivalent to PPA with the $M_h$-metric. By \cite{He2012ConvergenceAO}, we have:
  \begin{equation}
      \label{eq:proof_cp_1}
      R^k_h \leq \frac{1}{k}\|z^0_h - z^*_h\|^2_{M_{h}},\quad \forall z^*_h \in Z^*_h.
  \end{equation}
  The definition of $M_h$ gives us
 \[
 \|z^0_h - z^*_h\|^2_{M_{h}} = h^d \Big( \frac{1}{\mu } \|m^0_h - m^*_h \|^2_2  + \frac{1}{\tau }\|\phi^0_h - \phi^*_h\|^2_2  - 2\langle \phi^{0}_h-\phi^*_h, A_h(m^{0}_h-m^*_h) \rangle \Big),
  \]  
 where the last term of the right hand size can be bounded by applying the Peter–Paul inequality $ 2 \langle x,y\rangle  \leq \|x\|^2_2 / \epsilon + \epsilon \|y\|^2_2 $ with $x = -(\phi^{0}_h-\phi^*_h)$, $y = A_h(m^{0}_h-m^*_h)$ and $\epsilon = 1/\|A_h\|$: 
 \[ - 2\langle \phi^{0}_h-\phi^*_h, A_h(m^{0}_h-m^*_h) \rangle \leq \|A_h\|\cdot \|\phi^{0}_h-\phi^*_h\|^2_2 + \|A_h\|\cdot \|m^{0}_h-m^*_h\|^2_2.\]
 Given the parameter choice $\mu=\tau=1/(2\|A_h\|)$, we obtain the following bound
  \[
 \begin{aligned}
 \|z^0_h - z^*_h\|^2_{M_{h}} = & 2h^d \bigg( \|A_h\| \Big (\|m^0_h - m^*_h \|^2_2  + \|\phi^0_h - \phi^*_h\|^2_2 \Big )  -  \langle \phi^{0}_h-\phi^*_h, A_h(m^{0}_h-m^*_h) \rangle \bigg)\\
 \leq & 3 \|A_h\| h^d \Big (\|m^0_h - m^*_h \|^2_2  + \|\phi^0_h - \phi^*_h\|^2_2 \Big )\\
 \leq& 3 \|A_h\| \|z^0_h - z^*_h\|^2_{L^2}.
 \end{aligned}
  \]  
 The norm $\|A_h\|$ is the square root of the largest sigular-value of $(A_h)^*A_h$, which is the discrete Laplacian with grid step size $h$. By the  Gershgorin circle theorem~\cite{golub2012matrix}, $\sigma_{\text{max}}\big((A_h)^*A_h\big) \leq 4 \frac{d}{h^2}$ and, thus, $\|A_h\| \leq 2 \frac{\sqrt{d}}{h}$, which imples 
 \begin{equation}
 \label{eq:proof_cp_4}
     \|z^0_h - z^*_h\|^2_{M_{h}} \leq 6 \frac{\sqrt{d}}{h} \|z^0_h - z^*_h\|^2_{L^2}.
 \end{equation}
 Since we take zero as the initialization $z^0_h=0$, based on Assumption \ref{assume:cp2}, we have
 \begin{equation}
     \label{eq:proof_cp_5}
     \|z^0_h - z^*_h\|^2_{L^2} = \|z^*_h\|^2_{L^2} \leq C_1.
 \end{equation}
 Inequalities (\ref{eq:proof_cp_1}), (\ref{eq:proof_cp_4}) and (\ref{eq:proof_cp_5})
 imply
 \begin{equation}
     \label{eq:algo1_proof_main}
     R^k_h \leq \frac{1}{k}\|z^0_h - z^*_h\|^2_{M_{h}} \leq 6\sqrt{d} \frac{1}{k h}\|z^0_h - z^*_h\|^2_{L^2} \leq 6\sqrt{d} \frac{1}{k h} C_1.
 \end{equation}
 As long as $k > 6C_1 \frac{1}{\varepsilon}\frac{\sqrt{d}}{h}$, the stopping condition $R^k_h < \varepsilon$ is satisfied. Algorithm \ref{algo:primaldual} stops within $(6C_1 \frac{1}{\varepsilon}\frac{\sqrt{d}}{h}) \approx O(\frac{1}{\varepsilon}\frac{\sqrt{d}}{h})$ iterations, i.e., $K = O(\frac{1}{\varepsilon}\frac{\sqrt{d}}{h})$.
 
 \textbf{Step 2:} Analyzing how many iterations Algorithm \ref{algo:multigrid1} takes on level $l (l \geq 2)$.

Similar to (\ref{eq:proof_cp_1}) and  (\ref{eq:proof_cp_4}), on level $l$, it holds that 
\[R^k_{h_l} \leq \frac{1}{k}\|z^0_{h_l} - z^*_{h_l}\|^2_{M_{h_l}} \leq  \frac{6\sqrt{d}}{k h_l}\|z^0_{h_l} - z^*_{h_l}\|^2_{L^2}.\]
The distance in the right hand side can be bounded by
 \[
 \begin{aligned}
& \|z^0_{h_l} - z^*_{h_l}\|^2_{L^2} = \|\text{Interpolate }(z^{K_{l-1}}_{h_{l-1}}) - z^*_{h_l}\|^2_{L^2}\\ 
  \leq& \underbrace{2 \Big \|\text{Interpolate}(z^{K_{l-1}}_{h_{l-1}}) - \text{Interpolate}(z^*_{h_{l-1}})\Big \|^2_{L^2}}_{\mathrm{T1}}  + \underbrace{2 \Big \|\text{Interpolate}(z^*_{h_{l-1}}) - z^*_{h_l} \Big \|^2_{L^2}}_{\mathrm{T2}},
 \end{aligned}
 \]
 where the first term can be bounded by
 \[\mathrm{T1} = 2 \|\text{Interpolate }(z^{K_{l-1}}_{h_{l-1}} - z^*_{h_{l-1}})\|^2_{L^2} \leq 2 \|z^{K_{l-1}}_{h_{l-1}} - z^*_{h_{l-1}}\|^2_{L^2} \]because the ``interpolate" operator is linear and nonexpansive by 
Lemma \ref{lemma:inter_bdd}. 
Furthermore, due to Lemma \ref{lemma:epsilon} and the assumption that $\{\varepsilon_l\}_{l=1}^L \in \mathcal{T}_{\mathrm{1M}}$, we have 
\[\mathrm{T1} \leq 2 \|z^{K_{l-1}}_{h_{l-1}} - z^*_{h_{l-1}}\|^2_{L^2} \leq 2 C_2 (h_{l-1})^{r}.\]
The second term T2 can be bounded with Assumption \ref{assume:cp}:
$\mathrm{T2} \leq 2 C_2 (h_{l-1})^{r}$. Thus, for Algorithm \ref{algo:multigrid1}, we can obtain a better bound than (\ref{eq:proof_cp_5}):
\[\|z^0_{h_l} - z^*_{h_l}\|^2_{L^2} \leq \mathrm{T}_1 + \mathrm{T}_2 \leq 4 C_2 (h_{l-1})^{r} = 4 C_2 (2 h_{l})^{r} . \]
All the results above implies
\begin{equation}
    \label{eq:algo1m_proof_main}
    R^k_{h_l} \leq \frac{6\sqrt{d}}{k h_l} \|z^0_{h_l} - z^*_{h_l}\|^2_{L^2} \leq \frac{6\sqrt{d}}{k h_l} 4 C_2 (2 h_{l})^{r} = 24 \cdot 2^{r}C_2 \frac{ \sqrt{d}}{h_l^{1-r}}\frac{1}{k}.
\end{equation}
As long as $k > 24 \cdot 2^{r}C_2 \frac{ \sqrt{d}}{h_l^{1-r}}\frac{1}{\varepsilon_l}$, we have $R^k_{h_l} \leq \varepsilon_l$, i.e., the stopping condition on level $l$ is satisfied. Thus  the number of iterations on level $l$ can be bounded by $K_l \leq 24 \cdot 2^{r}C_2 \frac{ \sqrt{d}}{h_l^{1-r}}\frac{1}{\varepsilon_l} = O(\frac{1}{\varepsilon_l}\frac{\sqrt{d}}{h_l^{1-r}})$.
 Theorem \ref{theo:cp} is proved.
\end{proof}

Theorem \ref{lemma:cp} shows that the stopping tolerances in the sets $\mathcal{T}_{\mathrm{1}}$ and $\mathcal{T}_{\mathrm{1M}}$ are good choices. However, the conditions (\ref{eq:epsilon1}) and (\ref{eq:epsilon1m}) are intractable to check during we run the algorithms. Here we propose an explicit empirical formula for choosing $\{\varepsilon_l\}_{l=1}^L$:
\begin{equation}
    \label{eq:tol_formula_propose}
    \varepsilon_l = \varepsilon_L \times \Big( \frac{h_l}{h} \Big)^{\alpha },\quad \alpha \in \Re,
\end{equation}
where $\varepsilon_L$ is the stopping tolerance for the highest level $l=L$. With (\ref{eq:tol_formula_propose}), we provide the complexity analysis of Algorithms \ref{algo:primaldual} and \ref{algo:multigrid1} below.

\begin{theorem}
\label{theo:cp}
 Given $\rho^0,\rho^1,h,h_l = 2^{L-l}h$ and Assumptions \ref{assume:cp2},\ref{assume:cp}, if we take $0$ as the initialization, the following holds:
\begin{enumerate}
    \item  With $\varepsilon \in \mathcal{T}_{\mathrm{1}}$, the complexity of Algorithm \ref{algo:primaldual} is
    \begin{equation}
        \label{eq:thm_43_1}
        \mathcal{O}_{\mathrm{1}} = O\Big(\frac{1}{\varepsilon}\frac{d^{3/2}}{h^{d+1}}\Big).
    \end{equation}
    \item  With $\{\varepsilon_l\}_{l=1}^L$ chosen by (\ref{eq:tol_formula_propose}) and $\{\varepsilon_l\}_{l=1}^L \in \mathcal{T}_{\mathrm{1M}}$ and $L = 1 + \log_2(1/h)$, 
    the complexity of Algorithm \ref{algo:multigrid1} is 
    \begin{equation}
        \label{eq:thm_43_2}
        \mathcal{O}_{\mathrm{1M}} = 
\begin{cases}
O\Big(\frac{1}{\varepsilon_L} \frac{d^{3/2}}{h^{d+1-r}}\Big) + O\Big(\frac{d 2^d}{h^d}\Big), & \text{ if } \alpha > r - d - 1,\\
O\Big(\frac{1}{\varepsilon_L}\frac{d^{3/2}}{h^{d+1-r}}\log(\frac{1}{h})\Big) + O\Big(\frac{d 2^d}{h^d}\Big),&\text{ if } \alpha = r - d - 1,\\
O\Big(\frac{1}{\varepsilon_L}\frac{d^{3/2}}{h^{-\alpha}}\Big) + O\Big(\frac{d 2^d}{h^d}\Big), &\text{ if } \alpha < r - d - 1.
\end{cases}
    \end{equation}
\end{enumerate}
\end{theorem}

Theorem \ref{theo:cp} shows that, if $\{\varepsilon_l\}_{l=1}^L \in \mathcal{T}_{\mathrm{1M}}$, we should choose a large $\alpha$ to enjoy better complexity. A small $\alpha$ leads to small tolerance on lower levels $l < L$ and extra computational burdens on those levels. This is why $\mathcal{O}_{\mathrm{1M}} $ has a worse bound $h^{\alpha}$ as $\alpha < r - d - 1$. However, to guarantee $\{\varepsilon_l\}_{l=1}^L \in \mathcal{T}_{\mathrm{1M}}$ in practice, we have to choose smaller $\varepsilon_L$ if we choose larger $\alpha$. Then it will cost more calculations on the highest level, which hurts the advantage of multilevel methods.
Thus, we should choose a proper $\alpha$ and make balance between the computations on lower and higher levels. In Table \ref{tab:1m}, we will numerically test this point.

\begin{proof} 
 First, we consider the case where $p=1$ or $p=2$.
 
 For Algorithm \ref{algo:primaldual}, 
 the complexity is 
 ``iterations $\times$ single step complexity.'' In each step of Algorithm \ref{algo:primaldual}, the dominant calculation is computing $A_h m_h$ or $(A_h)^*\phi_h$~\cite{li2018parallel}, which has a complexity of $O(\frac{d}{h^d})$. Thus, 
 the total complexity of Algorithm \ref{algo:primaldual} is:
 \[\mathcal{O}_{\mathrm{1}} = O\Big(\frac{1}{\varepsilon}\frac{\sqrt{d}}{h}\Big) \times O\Big(\frac{d}{h^d}\Big) = O\Big(\frac{1}{\varepsilon}\frac{d^{3/2}}{h^{d+1}}\Big)
 .\]
 
 For Algorithm \ref{algo:multigrid1}, the complexity  is  the sum of two parts: iterations on all levels $\mathcal{O}_{\mathrm{1M,1}}$ and the interpolations between the levels $\mathcal{O}_{\mathrm{1M,2}}$. Let us first consider the former part. Similar to Algorithm \ref{algo:primaldual}, the complexity of level 1 is 
 $O(\frac{1}{\varepsilon_1}\frac{d^{3/2}}{h_1^{d+1}})$.
 The complexity of Level $l (2 \leq l \leq L)$ is 
 $O(\frac{1}{\varepsilon_l}\frac{d^{3/2}}{h_l^{d+1-r}})$. 
  Given $h_l = 2^{L-l}h$ and $L = 1 + \log_2(1/h)$, it holds that $h_l = 2^{-(l-1)}$. Due to (\ref{eq:tol_formula_propose}), we are able to obtain:
 \[
 \begin{aligned}
 \mathcal{O}_{\mathrm{1M,1}} = 
 & \sum_{l=2}^L O\bigg(\frac{1}{\varepsilon_l}\frac{d^{3/2}}{h_l^{d+1-r}}\bigg) + O\bigg(\frac{1}{\varepsilon_1}\frac{d^{3/2}}{h_1^{d+1}}\bigg) = O\bigg(\frac{h^\alpha }{\varepsilon_L}\bigg) \bigg( \sum_{l=2}^L \frac{d^{3/2}}{h_l^{d+1+\alpha -r}} + \frac{d^{3/2}}{h_1^{d+1+\alpha}} \bigg)\\
 = & O\bigg( \frac{1}{\varepsilon_L} \frac{d^{3/2}}{h^{d+1-r}}\bigg) \Big(\sum_{i=0}^{L-2}2^{-i(d+1+\alpha -r)} +  2^{-(d+1+\alpha )(L-1)} h^{-r} \Big)\\
  = & O\bigg( \frac{1}{\varepsilon_L} \frac{d^{3/2}}{h^{d+1-r}}\bigg) \Big(\sum_{i=0}^{L-1}2^{-i(d+1+\alpha -r)} \Big).
 \end{aligned}
 \]
 We consider three cases:
 \begin{itemize}
     \item If $\alpha > r - d -1$, we have $\sum_{i=0}^{L-1}2^{-i(d+1+\alpha-r)} < \infty$ is a constant.
     \item If $\alpha = r - d -1$, we have $\sum_{i=0}^{L-1}2^{-i(d+1+\alpha-r)} = L$. 
     \item If $\alpha < r - d -1$, we have $\sum_{i=0}^{L-1}2^{-i(d+1+\alpha-r)} = O(2^{L(r-d-1-\alpha)})$. 
 \end{itemize}
 Then $\mathcal{O}_{\mathrm{1M,1}}$ is obtained.
 
 
 Now let us consider the second part, the complexity of interpolations between the levels $l$ and $l+1$. Each node on level $l+1$ is obtained by no more than $2^d$ nodes, totally we have $O(d/h_{l+1}^d)$ nodes, so the  complexity of interpolation between the levels $l$ and $l+1$ is $O(d2^d/h_{l+1}^d)$. Summing them up, \[\mathcal{O}_{\mathrm{1M,2}} = O\bigg(\frac{d 2^d}{h^d}\bigg) \bigg(1 + \frac{1}{2^d} + \frac{1}{2^{2d}} + \cdots + \frac{1}{2^{(L-1)d}}\bigg) = O\bigg(\frac{d 2^d}{h^d}\bigg).\]
 
 For $p=\infty$, the dominant calculation in a single step includes two parts. One is computing $A_h m_h$ or $(A_h)^*\phi_h$, which is the same with the case of $p=1,2$. The other is calculating the $\ell_{\infty}$ shrinkage operator. By the Moreau decomposition~\cite{parikh2014proximal}, computing an $\ell_{\infty}$ shrinkage operator is equivalent with computing a projection onto an $\ell_1$ ball.  By~\cite{duchi2008efficient}, the complexity of the latter is $O(d)$. We need to project all the points $x \in \Omega^h$. In total, there are $O(N^d)=O(1/h^d)$ points, so the single step complexity is $O(\frac{d}{h^d})$. Following the above argument, we obtain the complexities of Algorithm \ref{algo:primaldual} and Algorithm \ref{algo:multigrid1} as $p=\infty$ have the same asymptotic rate as $p=1,2$.
\end{proof}

 \subsection{Analysis of Algorithms \ref{algo:pdhg} and \ref{algo:multigrid2}} 

Similar to Lemma \ref{lemma:epsilon}, we establish the conditions of stopping tolerances for Algorithms \ref{algo:pdhg} and \ref{algo:multigrid2}.

\begin{lemma}
\label{lemma:epsilon_pdhg}
Given $\rho^0,\rho^1,h$, we assume $h_l = 2^{L-l}h$. 
\begin{enumerate}
    \item There exists $\tilde{\varepsilon}$, as long as we run Algorithm \ref{algo:pdhg} with $\varepsilon \leq \tilde{\varepsilon}$, we have 
    \begin{equation}
    \label{eq:lemma41}
        \|y^K_{h}-y^*_{h}\|^2_{L^2} \leq C_2 h^{\nu},\quad \text{for some }y^*_{h} \in Y^*_{h}.
    \end{equation}
    \item Suppose we run Algorithm \ref{algo:multigrid2} with arbitrary tolerances $\{\varepsilon_i\}_{i=1}^{l-1}$ on the lowest $l-1$ levels. There exists a threshold $\tilde{\varepsilon}_{l}$ depending on $\{\varepsilon_i\}_{i=1}^{l-1}$, as long as we set $\varepsilon_{l} \leq \tilde{\varepsilon}_{l}\Big(\{\varepsilon_i\}_{i=1}^{l-1}\Big)$ on level $l$, we have
    \begin{equation}
    \label{eq:lemma42}
        \|y^{K_l}_{h_{l}}-y^*_{h_{l}}\|^2_{L^2} \leq C_2 (h_{l})^{\nu},\quad \text{for some }y^*_{h_{l}} \in Y^*_{h_{l}}.
    \end{equation}
\end{enumerate}
\end{lemma}

\begin{proof}
Define \[\tilde{M} = \begin{bmatrix} I/\mu & I \\I & I/\tau  \end{bmatrix}.\]Similar to Algorithm \ref{algo:primaldual}, Algorithm \ref{algo:pdhg} is equivalent to PPA with the $\tilde{M}$-metric. Then we follow the same proof as that of Lemma \ref{lemma:epsilon} and obtain Lemma \ref{lemma:epsilon_pdhg}.
\end{proof}

Define the set of good stopping tolerances for Algorithms \ref{algo:pdhg} and \ref{algo:multigrid2}:
\begin{align}
 \mathcal{T}_{\mathrm{2}} = & \bigg \{\varepsilon \bigg | \varepsilon \leq \tilde{\varepsilon}\bigg \}\label{eq:epsilon2},\\ \mathcal{T}_{\mathrm{2M}} =  & \bigg  \{\{\varepsilon_l\}_{l=1}^L \bigg  | \varepsilon_1 \leq \tilde{\varepsilon}_1, ~\varepsilon_2 \leq \tilde{\varepsilon}_2(\varepsilon_1),~\cdots,~\varepsilon_L\leq \tilde{\varepsilon}_{L}\Big(\{\varepsilon_l\}_{l=1}^{L-1}\Big)\bigg  \}\label{eq:epsilon2m}.
\end{align}

\begin{theorem}
\label{lemma:pdhg}
 Given $\rho^0,\rho^1,h,L\geq 2$ and $h_l = 2^{L-l}h$ and Assumptions \ref{assume:pdhg2},\ref{assume:pdhg}, if we take $0$ as the initialization, then the following holds:
\begin{enumerate}
    \item  With $\varepsilon \in \mathcal{T}_{\mathrm{2}}$, Algorithm \ref{algo:pdhg} takes $O(\frac{1}{\varepsilon})$ iterations to stop. 
    \item With $\{\varepsilon_l\}_{l=1}^L \in \mathcal{T}_{\mathrm{2M}}$, Algorithm \ref{algo:multigrid2} takes
    $O(\frac{1}{\varepsilon_l}{h_l}^{\nu})$
    iterations on level $l$ ($l \geq 2$). 
\end{enumerate}
\end{theorem}

\begin{proof}
Firstly, we follow the same argument of step 1 in the proof of Theorem \ref{lemma:cp}. Substituting $-A_h$ with $I$ and $M_h$ with $\tilde{M}$, we obtain the following inequality similar to (\ref{eq:algo1_proof_main}):
\[G^k_h \leq \frac{1}{k}\|y^0_h-y^*_h\|^2_{\tilde{M}} \leq \frac{3}{k}\|y^0_h - y^*_h\|^2_{L^2} =  \frac{3}{k}\|y^*_h\|^2_{L^2} \leq \frac{3C_3}{k}.\]
As long as $k \geq \frac{3C_3}{\varepsilon}$, the stopping condition $G^k_h \leq \varepsilon$ is satisfied.  Algorithm \ref{algo:pdhg} takes $K \leq \frac{3C_1}{\varepsilon} = O(1/\varepsilon)$ iterations to stop. 

Similar to Algorithm \ref{algo:pdhg}, Algorithm \ref{algo:multigrid2} takes $O(\frac{1}{\varepsilon_1})$ iterations for level $l=1$. For level $l \geq 2$, we follow the same argument of step 2 in the proof of Theorem \ref{lemma:cp} and obtain
\[
     G^k_{h_l} \leq \frac{3}{k } \|y^0_{h_l} - y^*_{h_l}\|^2_{L^2} \leq \frac{3}{k} 4 C_4 (2 h_{l})^{\nu} = 12 \cdot 2^{r}C_4 (h_l)^{\nu}\frac{1}{k}.
\]
As long as $k \geq 12 \cdot2^{r}C_4 (h_l)^{\nu}\frac{1}{\varepsilon_l}$, the stopping condition $G^k_{h_l} \leq \varepsilon$ is satisfied.  Algorithm \ref{algo:multigrid2} takes $K_l \leq 12 \cdot2^{r}C_4 (h_l)^{\nu}\frac{1}{\varepsilon_l} = O(\frac{1}{\varepsilon_l}{h_l}^{\nu})$ iterations to stop. 
\end{proof}

For Algorithm \ref{algo:multigrid2}, we choose the same stopping tolerance rule with that of Algorithm \ref{algo:multigrid1} (\ref{eq:tol_formula_propose}) and obtain the following complexity analysis.

\begin{theorem}
\label{theo:pdhg}
 Given $\rho^0,\rho^1,h, h_l = 2^{L-l}h$ and Assumptions \ref{assume:pdhg2},\ref{assume:pdhg}, if we take $0$ as the initialization, the following holds:
\begin{enumerate}
    \item  With $\varepsilon \in \mathcal{T}_{\mathrm{2}}$, the complexity of Algorithm \ref{algo:pdhg} is
    \begin{equation}
        \label{eq:thm_46_1}
        \mathcal{O}_{\mathrm{2}} = O\Big(\frac{1}{\varepsilon}\frac{d}{h^{d }}\log(\frac{1}{h})\Big).
    \end{equation}
    \item  With $\{\varepsilon_l\}_{l=1}^L$ taken by (\ref{eq:tol_formula_propose}) and $\{\varepsilon_l\}_{l=1}^L \in \mathcal{T}_{\mathrm{2M}}$ and $L=1+\log_2(1/h)$, the complexity of Algorithm \ref{algo:multigrid2} is 
\begin{equation}
    \label{eq:thm_46_2}
    \mathcal{O}_{\mathrm{2M}} = 
    \begin{cases}
    O\Big(\frac{1}{\varepsilon_L} \frac{d}{h^{d-\nu}}\log(\frac{1}{h}))\Big)+O\Big(\frac{d 2^d}{h^d}\Big), & \text{ if } \alpha > \nu - d,\\
 O\Big(\frac{1}{\varepsilon_L} \frac{d}{h^{d-\nu}}\log^2(\frac{1}{h}))\Big)+O\Big(\frac{d 2^d}{h^d}\Big), & \text{ if } \alpha = \nu - d,\\
 O\Big(\frac{1}{\varepsilon_L} \frac{d}{h^{-\alpha}}\log(\frac{1}{h}))\Big)+O\Big(\frac{d 2^d}{h^d}\Big), & \text{ if }  \alpha < \nu - d.
    \end{cases}
\end{equation}
\end{enumerate}
\end{theorem}

\begin{proof} 
In each step of Algorithm \ref{algo:pdhg}, the dominant calculation is conducting a $d$ dimensional FFT on $\bar{\varphi}^k_h$~\cite{jacobs2018solving}, which have complexity of $O(N^d\log(N^d)$. Since $N=1/h$, the complexity is $O( \frac{d}{h^d}\log(\frac{1}{h}))$~\cite{duhamel1990fast}. Then the complexity of Algorithm \ref{algo:pdhg} is:
 \[\mathcal{O}_{\mathrm{2}} = O\Big(\frac{1}{\varepsilon}\Big) \times O\bigg( \frac{d}{h^d}\log(\frac{1}{h})\bigg) = O\bigg(\frac{1}{\varepsilon}\frac{d}{h^{d}}\log(\frac{1}{h})\bigg).\]
 
 Thus, the complexity of Algorithm \ref{algo:multigrid2} on level $l=1$ is $O(\frac{1}{\varepsilon_1}\frac{d}{h_1^{d}}\log(\frac{1}{h_1}))$, and, on higher levels $(2 \leq l \leq L)$, is $O(\frac{1}{\varepsilon_l}\frac{d}{h_l^{d -\nu}}\log(\frac{1}{h_l}))$.
 Therefore, we have
 \[
 \begin{aligned}
\mathcal{O}_{\mathrm{2M,1}} = & \sum_{l=2}^L O\bigg(\frac{1}{\varepsilon_l}\frac{d}{h_l^{d -\nu}}\log(\frac{1}{h_l}))\bigg) + O\bigg(\frac{1}{\varepsilon_1}\frac{d}{h_1^{d}}\log(\frac{1}{h_1})\bigg)\\ 
 =& \sum_{l=2}^L O\bigg(\frac{h^\alpha}{\varepsilon_L}\frac{d}{h_l^{d+\alpha -\nu}}\log(\frac{1}{h_l}))\bigg) + O\bigg(\frac{h^\alpha}{\varepsilon_L}\frac{d}{h_1^{d+\alpha}}\log(\frac{1}{h_1})\bigg)\\ 
 = & O\bigg(\frac{1}{\varepsilon_L} \frac{d}{h^{d-\nu}}\log(\frac{1}{h}))\bigg) \Big(\sum_{i=0}^{L-2}2^{-i(d+\alpha-\nu)} + 2^{-(d+\alpha)(L-1)} h^{-\nu} \Big)\\
 = & O\bigg(\frac{1}{\varepsilon_L} \frac{d}{h^{d-\nu}}\log(\frac{1}{h}))\bigg) \Big(\sum_{i=0}^{L-1}2^{-i(d+\alpha-\nu)} \Big).
 \end{aligned}
 \]
 Then following the same proof as that for Theorem \ref{theo:cp}, we obtain (\ref{eq:thm_46_2}). 
\end{proof}

\subsection{Summary of complexities} Tables \ref{tab:1} and \ref{tab:2} summarize the complexities.
Let $N=1/h$, Table \ref{tab:1} can be directly derived from Theorems \ref{theo:cp} and \ref{theo:pdhg}. 
In the case of $d=2$ (2D case), we choose typical parameters (tolerances in Section \ref{sec:prac-tol}, $r=2,\nu=1$) and compare Algorithms \ref{algo:primaldual}, \ref{algo:pdhg}, \ref{algo:multigrid1} and \ref{algo:multigrid2} with other EMD algorithms~\cite{Ling,bassetti2018computation} in Table \ref{tab:2}. By \cite{Ling}, their algorithm has complexity of  $O((N^d)^2)$. As $d=2$, it is $O(N^4)$. The algorithm in \cite{bassetti2018computation}  constructs a graph and solves the uncapacitated minimum cost flow problem on the graph. The worst case complexity is $O\big(|V|\log(|V|)(|V|\log(|V|)+|E|)\big)$, where $|V|$ is the number of nodes in the created graph and $|E|$ is the number of edges. As $p=1$ or $p=\infty$, $|V|=O(N^2),|E|=O(N^2)$, the complexity is $O(N^{4}\log^2(N))$; as $p=2$, $|V|=O(N^2),|E|=O(N^4)$, the complexity is $O(N^{6}\log(N))$. 

\begin{table}[ht]
\caption{Complexities of Algorithms \ref{algo:primaldual}, \ref{algo:pdhg}, \ref{algo:multigrid1} and \ref{algo:multigrid2}. The parameters $r,\nu$ depend on the interpolation accuracy (Assumptions \ref{assume:cp}, \ref{assume:pdhg}). We take $\varepsilon_l$ as (\ref{eq:tol_formula_propose}) and $\alpha\geq r-d-1$ for Alg. \ref{algo:multigrid1}, $\alpha \geq \nu -d$ for Alg. \ref{algo:multigrid2}.}
\label{tab:1}
\centering
\begin{tabular}{l|l}
\hline
 \multicolumn{2}{c}{$p=1,2,\infty$}                                        \\ \hline
 Algorithm \ref{algo:primaldual}~\cite{li2018parallel} & $O(d^{3/2}N^{d+1}\log(N)/\varepsilon)$  \\ \hline
 Algorithm \ref{algo:pdhg}~\cite{jacobs2018solving} & $O(d N^{d}\log^2(N)/\varepsilon)$ \\ \hline
 Algorithm \ref{algo:multigrid1} &   $O(d^{3/2}N^{d+1-r}\log(N)/\varepsilon_L)+O(d 2^d N^d)$   \\\hline   
 Algorithm \ref{algo:multigrid2} &    $O(d N^{d-\nu} \log^2(N)/\varepsilon_L)+O(d 2^d N^d)$        \\ \hline
\end{tabular}
\end{table}

\begin{table}[ht]
\caption{Complexities in the 2D case ($d=2$). Algorithms \ref{algo:primaldual},\ref{algo:pdhg},\ref{algo:multigrid1},\ref{algo:multigrid2} choose tolerances described in Section \ref{sec:prac-tol}.
}
\label{tab:2}
\centering
\begin{tabular}{l|l|l|l}
\hline
  & $p=1$ & $p=2$ & $p=\infty$ \\ \hline
Tree-EMD~\cite{Ling} &  $O(N^{4})$    &   -  &   -    \\ \hline
Min-cost flow\cite{bassetti2018computation}\footnotemark &  $O(N^{4}\log^2(N))$   &    $O(N^{6}\log(N))$  &  $O(N^{4}\log^2(N))$      \\ \hline
Algorithm \ref{algo:primaldual}~\cite{li2018parallel}  &  $O(N^{6})$   &  $O(N^{6})$    &   $O(N^{6})$     \\ \hline
Algorithm \ref{algo:pdhg}~\cite{jacobs2018solving}  &  $O( N^{4} \log(N))$   &   $O( N^{4} \log(N))$  &  $O(N^{4} \log(N))$     \\ \hline
Algorithm \ref{algo:multigrid1}    &  \multirow{1}{*}{$O( N^{3}\log(N))$}    &  \multirow{1}{*}{$O( N^{3}\log(N))$}   &    \multirow{1}{*}{$O(N^{3}\log(N))$}   \\
\hline
Algorithm \ref{algo:multigrid2}  &   \multirow{1}{*}{$O( N^{3} \log^2(N))$}  &  \multirow{1}{*}{$O( N^{3} \log^2(N))$}   &   \multirow{1}{*}{$O( N^{3} \log^2(N))$}    \\ 
\hline
\end{tabular}
\end{table}

\footnotetext{The complexity of solving the minimum cost flow problem is the upper bound for the worst case. In practice, their algorithm has better performance than the theoretical bound. Numerical results are reported in Table \ref{tab:final}.}

\section{Numerical validation of the assumptions}\label{section:assumes}
We numerically validated Assumptions \ref{assume:cp2}, \ref{assume:cp}, \ref{assume:pdhg2} and \ref{assume:pdhg} in the case of dimension $d=2$ and $p \in \{1,2,\infty\}$. 

\subsection{Quantitative validation}
\label{sec:assumes}
Table \ref{tab:assume2} reports our results.
Since the EMD generally does not have a closed-form solution in the 2D case, we numerically estimate $z^*_{h_l},y^*_{h_l} (1\leq l \leq L)$ to validate the assumptions.
By Theorem 1 in \cite{li2018parallel}, $z^k_{h_l}\to z^*_{h_l}$ as $k\to\infty$ for all $l$. Consequently, as long as the stopping tolerance $\varepsilon$ is small enough, we could use $z^K_{h_l}$ obtained by Algorithm \ref{algo:multigrid1} to estimate $z^*_{h_l}$. In this section, we set $\varepsilon=10^{-8}$. The same can be done to estimate $y^*_{h_l}$. We estimated $z^*_{h_l},y^*_{h_l} (1\leq l \leq L)$ on the instances in the DOTmark dataset~\cite{schrieber2017dotmark} and calculated averages over the instances of the four values: $\|z^*_{h_l}\|^2_{L^2}$, $\|y^*_{h_l}\|^2_{L^2}$, $\|\mathrm{Interpolate }(z^*_{h_{l-1}}) - z^*_{h_{l}}\|^2_{L^2}$, $\|\mathrm{Interpolate }(y^*_{h_{l-1}}) - y^*_{h_{l}}\|^2_{L^2}$.

Table \ref{tab:assume2} shows that $\|z^*_{h_l}\|^2_{L^2}$ and $\|y^*_{h_l}\|^2_{L^2}$ are both bounded by a constant independent of grid step size $h_l$. 
Furthermore, $\|\mathrm{Interpolate }(z^*_{h_{l-1}}) - z^*_{h_{l}}\|^2_{L^2}\approx O((h_l)^2)$ and $\|\mathrm{Interpolate }(y^*_{h_{l-1}}) - y^*_{h_{l}}\|^2_{L^2}\approx O(h_l)$. The above conclusions directly validate Assumptions \ref{assume:cp2}, \ref{assume:cp}, \ref{assume:pdhg2}, \ref{assume:pdhg} and $r \approx 2, \nu \approx 1$.

\begin{table}[t]
\caption{Validation of the assumptions}
\label{tab:assume2}
\centering
\begin{tabular}{c|c|c|c|c|c}
\hline
      & $l=1$ & $l=2$  & $l=3$ & $l=4$ & $l=5$ \\
      \cline{2-6}
     & $h_l=1/32$ & $h_l=1/64$  & $h_l=1/128$ & $h_l=1/256$ & $h_l=1/512$ \\\hline\hline
     \multicolumn{6}{c}{Assumption \ref{assume:cp2}: $\|z^*_{h_l}\|^2_{L^2}$}              \\ \hline
$p=1$&0.101&0.094&0.091&0.089&0.089\\\hline
$p=2$&0.061&0.057&0.056&0.055&0.055\\\hline
$p=\infty$&0.056&0.053&0.052&0.051&0.051\\\hline
\hline
\multicolumn{6}{c}{Assumption \ref{assume:cp}: $\|\mathrm{Interpolate }(z^*_{h_{l-1}}) - z^*_{h_{l}}\|^2_{L^2}$}              \\ \hline
$p=1$ & $1.35 \times 10^{-3}$ & $4.15 \times 10^{-4}$ & $1.08 \times 10^{-4}$ & $2.57 \times 10^{-5}$ & $3.59 \times 10^{-6}$ \\ \hline
$p=2$ & $5.34 \times 10^{-4}$ &  $1.71 \times 10^{-4}$ &  $4.40 \times 10^{-5}$ & $9.67 \times 10^{-6}$ & $1.97 \times 10^{-6}$ \\ \hline
$p = \infty$ & $8.79 \times 10^{-4}$ & $2.93 \times 10^{-4}$ & $7.08 \times 10^{-5}$ & $2.07 \times 10^{-5}$ & $6.77 \times 10^{-6}$ \\ \hline
\hline
\multicolumn{6}{c}{Assumption \ref{assume:pdhg2}: $\|y^*_{h_l}\|^2_{L^2}$}              \\ \hline
$p=1$&2.019&1.984&1.972&1.967&1.965\\\hline
$p=2$&1.015&1.003&0.997&0.995&0.993\\\hline
$p=\infty$&0.965&0.960&0.963&0.967&0.972\\\hline
\hline
\multicolumn{6}{c}{Assumption \ref{assume:pdhg}: $\|\mathrm{Interpolate }(y^*_{h_{l-1}}) - y^*_{h_{l}}\|^2_{L^2}$}              \\ \hline
$p=1$ & $2.55\times10^{-1}$ & $1.33\times10^{-1}$ & $4.42\times10^{-2}$ &$1.01\times10^{-2}$ & $3.30\times10^{-3}$ \\\hline
$p=2$ & $6.05\times10^{-2}$ &$2.89\times10^{-2}$&$1.31\times10^{-2}$ &$4.68\times10^{-3}$ & $1.14\times10^{-3}$ \\\hline
$p=\infty$ & $8.18\times10^{-2}$ &$4.72\times10^{-2}$&$2.12\times10^{-2}$&$8.84\times10^{-3}$&$3.48\times10^{-3}$\\\hline
\end{tabular}
\end{table}

\subsection{Visualization of the non-uniqueness cases}

\paragraph{Visualization of $z^*_{h_l}$} The solution set $Z^*_{h_l}$ may have multiple solutions on each level. This phenomena is also studied in \cite{li2018parallel} with only one level. What we want to show in this paragraph is that, for each coarse level solution $z^*_{h_l} \in Z^*_{h_l}$, there is a finer level solution $z^*_{h_{l+1}} \in Z^*_{h_{l+1}}$ that is close to $z^*_{h_l}$. Here we set $p=1$. 
Figures \ref{fig:assume1m} and \ref{fig:assume1phi} illustrate this point. 
First, we consider the primal solution  $m^*_{h_l}$. With different initializations, we obtain two different optimal $m^*_{h_l}$s on level $l=1$: Figures \ref{fig:visualize_1_8} and \ref{fig:visualize_2_8}. With the results in Figures \ref{fig:visualize_1_8} and \ref{fig:visualize_2_8} as initializations, we obtain the solutions $m^*_{h_l}$ on level $2$: Figures \ref{fig:visualize_1_16} and \ref{fig:visualize_2_16}. The flux in Figure \ref{fig:visualize_1_16} is close to that in Figure \ref{fig:visualize_1_8}; Figure \ref{fig:visualize_2_16} is close to Figure \ref{fig:visualize_2_8}. Thus, Assumption \ref{assume:cp} is meaningful when there are multiple solutions on each level: for a solution $m^*_{h_l}$ on level $l$, there is a solution $m^*_{h_{l+1}}$ on level $l+1$ similar to $m^*_{h_l}$. 
Secondly, we consider the dual solution $\phi^*_{h_l}$. As $p=1$, $\phi^*_{h_l}$ is unique up to a constant for each level $l$. The $\phi^*_{h_l}$ on level $l$ is close to that on the finer level $\phi^*_{h_{l+1}}$. Figure \ref{fig:assume1phi} demonstrates this point.

\begin{figure}
\centering
\begin{tabular}{ccc}
\hspace{-10mm}
\subfigure[][
\parbox{0.2\textwidth}
{One $m^*_{h_1}$: ~$8\times8$. Energy: $1.000113$.}]{
\includegraphics[width=0.35\textwidth]{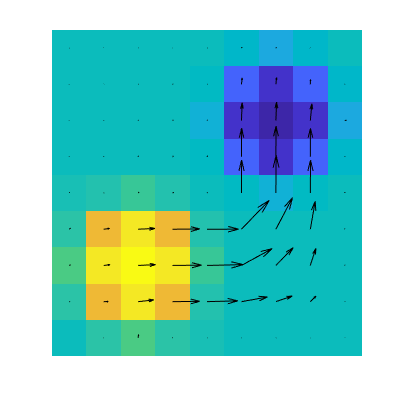}\label{fig:visualize_1_8}}
&
\hspace{-10mm}
\subfigure[][
\parbox{0.2\textwidth}
{One $m^*_{h_2}$: $16\times16$. Energy: $1.000127$.}]{
\includegraphics[width=0.35\textwidth]{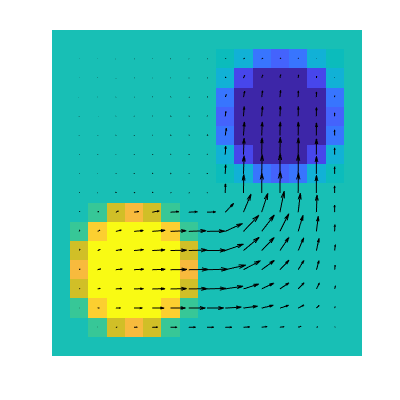}\label{fig:visualize_1_16}}
&
\hspace{-10mm}
\subfigure[][
\parbox{0.2\textwidth}
{One $m^*_{h_3}$: $32\times32$. Energy: $0.999764$} ]{
\includegraphics[width=0.35\textwidth]{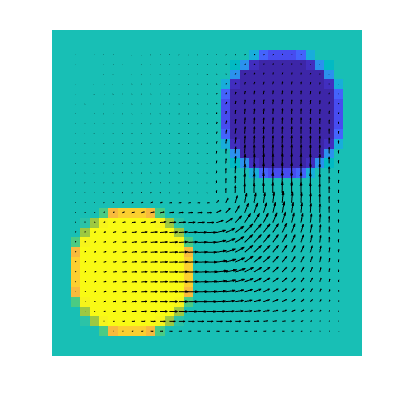}\label{fig:visualize_1_32}}\\
\hspace{-10mm}
\subfigure[][\parbox{0.2\textwidth}{Another $m^*_{h_1}$: ~$8\times8$. Energy: $1.000113$.}]{
\includegraphics[width=0.35\textwidth]{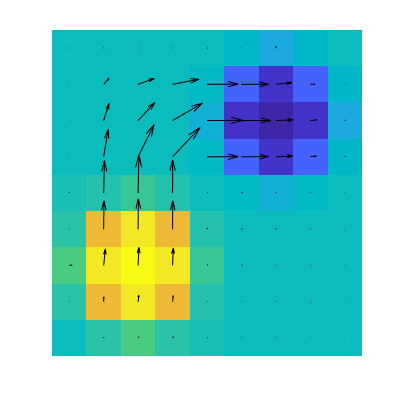}\label{fig:visualize_2_8}}
&
\hspace{-10mm}
\subfigure[][\parbox{0.25\textwidth}{Another $m^*_{h_2}$: $16\times 16$. Energy: $1.000127$.}]{
\includegraphics[width=0.35\textwidth]{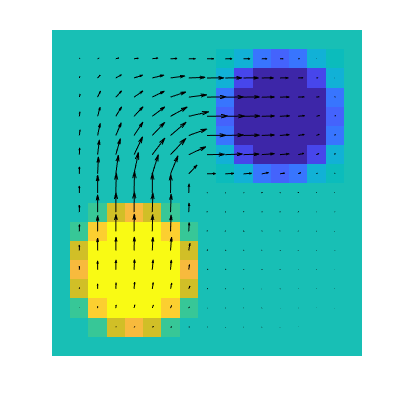}\label{fig:visualize_2_16}}
&
\hspace{-10mm}
\subfigure[][\parbox{0.25\textwidth}{Another $m^*_{h_3}$: $32\times32$. Energy: $0.999764$.}]{
\includegraphics[width=0.35\textwidth]{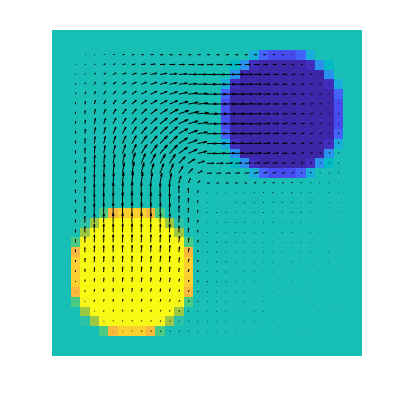}\label{fig:visualize_2_32}}
\end{tabular}
\caption{Visualization of Assumption \ref{assume:cp}: the black flux represents $m_{h_l}:\Omega^{h_l} \to \Re^2$, the two circles represent $\rho_{h_l}^0,\rho^1_{h_l}$ respectively. There are multiple solutions on each level. For every solution $m^*_{h_l}$ on level $l$, there is a solution $m^*_{h_{l+1}}$ on the finer level $l+1$ that is similar to $m^*_{h_l}$.
}
\label{fig:assume1m}

\vspace{5mm}
\begin{tabular}{ccc}
\hspace{-10mm}
\subfigure[][\parbox{0.15\textwidth}{$\phi^*_{h_1}$: $8\times8$.}]{
\includegraphics[width=0.35\textwidth]{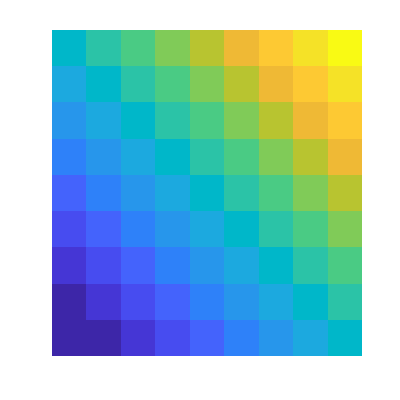}\label{fig:visualize_1_8phi}}
&
\hspace{-10mm}
\subfigure[][\parbox{0.15\textwidth}{$\phi^*_{h_2}$: $16\times16$.}]{
\includegraphics[width=0.35\textwidth]{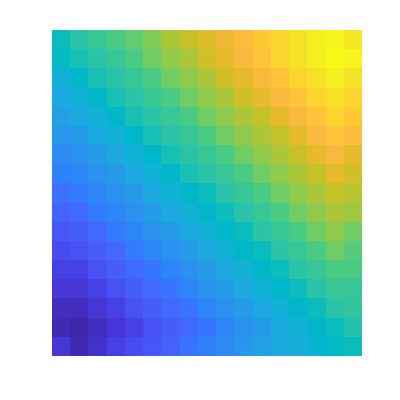}\label{fig:visualize_1_16phi}}
&
\hspace{-10mm}
\subfigure[][\parbox{0.15\textwidth}{$\phi^*_{h_3}$: $32\times32$.}]{
\includegraphics[width=0.35\textwidth]{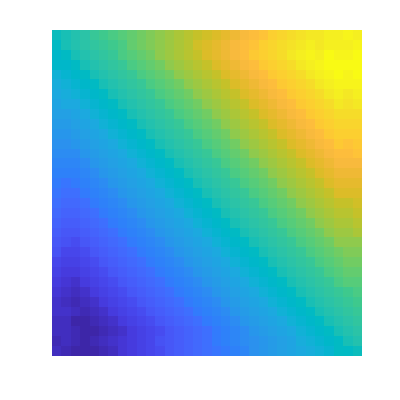}\label{fig:visualize_1_32phi}}
\end{tabular}
\caption{Visualization of Assumption \ref{assume:cp}: dual solution (Kantorovich potential) $\phi^*_{h_l}: \Omega^{h_l} \to \Re$ on each level. The dual solution $\phi^*_{h_{l+1}}$ on level $l+1$ is close to $\phi^*_{h_l}$ on the coarser level $l$.
}
\label{fig:assume1phi}
\end{figure}

\begin{figure}
\centering
\begin{tabular}{ccc}
\hspace{-10mm}
\subfigure[][\parbox{0.2\textwidth}{One $m^*_{h_1}$:~ $8\times8$. Energy: $0.941176$.}]{
\includegraphics[width=0.35\textwidth]{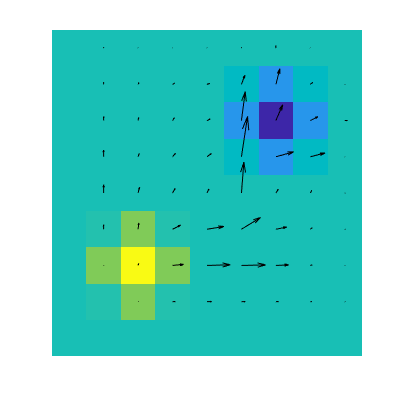}\label{fig:visualize3_1_8}}
&
\hspace{-10mm}
\subfigure[][\parbox{0.2\textwidth}{One $m^*_{h_2}$: $16\times16$. Energy: $0.969697$.}]{
\includegraphics[width=0.35\textwidth]{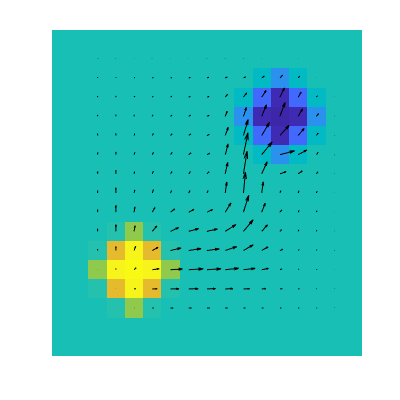}\label{fig:visualize3_1_16}}
&
\hspace{-10mm}
\subfigure[][\parbox{0.2\textwidth}{One $m^*_{h_3}$: $32\times32$. Energy: $0.984615$.}]{
\includegraphics[width=0.35\textwidth]{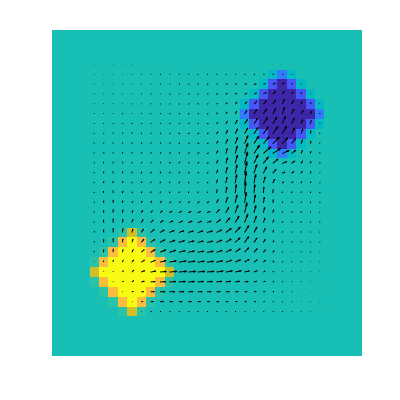}\label{fig:visualize3_1_32}}\\
\hspace{-10mm}
\subfigure[][\parbox{0.2\textwidth}{Another $m^*_{h_1}$: $8\times8$. Energy: $0.941176$.}]{
\includegraphics[width=0.35\textwidth]{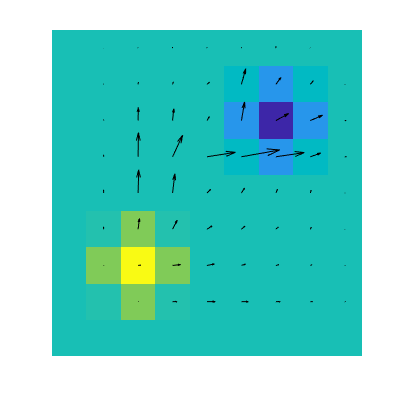}\label{fig:visualize3_2_8}}
&
\hspace{-10mm}
\subfigure[][\parbox{0.25\textwidth}{Another $m^*_{h_2}$: $16\times16$. Energy: $0.969697$.}]{
\includegraphics[width=0.35\textwidth]{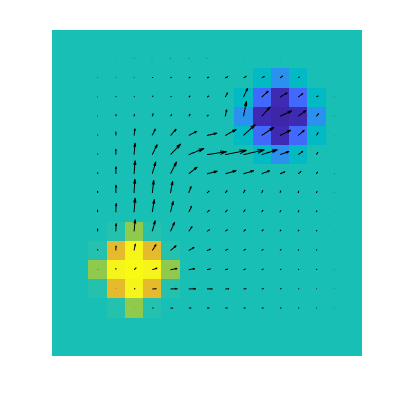}\label{fig:visualize3_2_16}}
&
\hspace{-10mm}
\subfigure[][\parbox{0.25\textwidth}{Another $m^*_{h_3}$: $32\times32$. Energy: $0.984615$.}]{
\includegraphics[width=0.35\textwidth]{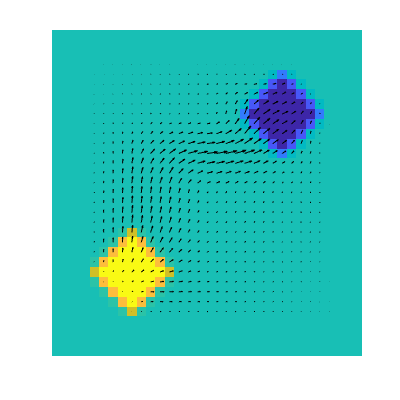}\label{fig:visualize3_2_32}}
\end{tabular}
\caption{Visualization of Assumption \ref{assume:pdhg}: the black flux represents $m_{h_l}:\Omega^{h_l} \to \Re^2$, the the two circles represent $\rho_{h_l}^0,\rho^1_{h_l}$ respectively. There are multiple solutions on each level. For every solution $m^*_{h_l}$ on level $l$, there is a solution $m^*_{h_{l+1}}$ on level $l+1$ that is close to $m^*_{h_l}$.
}
\label{fig:assume3m}

\begin{tabular}{ccc}
\hspace{-10mm}
\subfigure[][\parbox{0.2\textwidth}{Dual solution $\varphi^*_{h_1}$ on $8\times8$ grid.}]{
\includegraphics[width=0.35\textwidth]{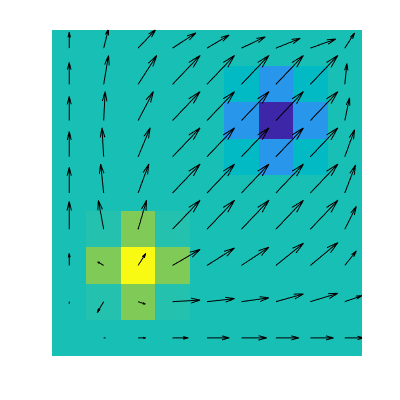}\label{fig:visualize3_1_8p}}
&
\hspace{-10mm}
\subfigure[][\parbox{0.2\textwidth}{Dual solution $\varphi^*_{h_2}$ on  $16\times16$ grid.}]{
\includegraphics[width=0.35\textwidth]{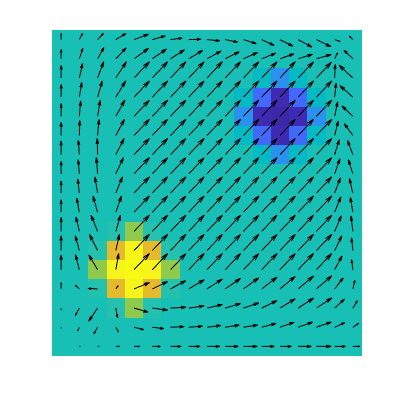}\label{fig:visualize3_1_16p}}
&
\hspace{-10mm}
\subfigure[][\parbox{0.2\textwidth}{Dual solution $\varphi^*_{h_3}$ on  $32\times32$ grid.}]{
\includegraphics[width=0.35\textwidth]{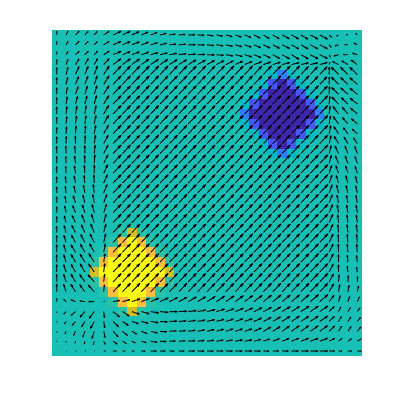}\label{fig:visualize3_1_32p}}
\end{tabular}
\caption{Visualization of Assumption \ref{assume:pdhg}: dual solution $\varphi^*_{h_l}: \Omega^{h_l} \to \Re^2$ on each level. The dual solution $\varphi^*_{h_{l+1}}$ on level $l$ is close to $\varphi^*_{h_l}$ on level $l+1$.
}
\label{fig:assume3p}
\end{figure}

\paragraph{Visualization of $y^*_{h_l}$} We visualize $y^*_{h_l}$ in a similar way. We set $p=1$ and get the results in
Figures \ref{fig:assume3m} and \ref{fig:assume3p}. 
Figure \ref{fig:assume3m} shows that for a solution $m^*_{h_l}$ on level $l$, there is a solution $m^*_{h_{l+1}}$ on level $l+1$, which is similar to $m^*_{h_l}$. 
On this specific numerical example, the dual variable $\varphi^*_{h_l}$ is unique for each level $l$. Figure \ref{fig:assume3p} demonstrates that $\varphi^*_{h_l}$ on level $l$ is close to $\varphi^*_{h_{l+1}}$ on level $l+1$.

\section{Numerical results}\label{section:5}
In this section, we numerically study why  and how much our Algorithms \ref{algo:multigrid1} and \ref{algo:multigrid2} speed up Algorithms \ref{algo:primaldual} and \ref{algo:pdhg}. The conclusions in Theorems \ref{theo:cp} and \ref{theo:pdhg} are validated. Moreover, we compare our algorithms with other EMD solvers \cite{Ling,bassetti2018computation,li2018parallel,jacobs2018solving}.
We implemented Algorithms \ref{algo:multigrid1} and \ref{algo:multigrid2} for $d=2$ in MATLAB. 
All the experiments were conducted on a single CPU (Intel i7-2600 CPU @ 3.40GHz).

\subsection{The effect of multilevel initialization}
\label{sec:cats}

In this subsection, we study why multilevel initialization helps speed up Algorithms \ref{algo:primaldual} and \ref{algo:pdhg}. 
All the results are obtained on the ``cat'' example,
which is also used as a benchmark in \cite{li2018parallel,jacobs2018solving},
and $p=1$.
For Algorithm \ref{algo:multigrid1}, we choose $\varepsilon_l = 10^{-6}$ for all levels. For Algorithm \ref{algo:multigrid2}, we choose $\varepsilon_l = 10^{-5}$ for all levels.
We report the results in Table \ref{tab:pd_comparel}.

\begin{table}[ht]
\centering
\caption{The effect of level number $L$. ``Iters'' is the number of iterations, and ``Time'' is in second.}
\label{tab:pd_comparel}
\begin{tabular}{c|c|c|c|c|c|c|c|c|c}
\hline
\multirow{2}{*}{Number} & \multicolumn{8}{c|}{Calculation cost on each level}                                                                                                               & \multirow{2}{*}{Total} \\ \cline{2-9}
                  & \multicolumn{2}{c|}{$64\times64$} & \multicolumn{2}{c|}{$128\times128$} & \multicolumn{2}{c|}{$256\times256$} &   \multicolumn{2}{c|}{$512\times512$} &                                 \\ \cline{2-9}
           of levels       & Iters   & Time    & Iters   & Time    & Iters   & Time  & Iters   & Time   &  time                                 \\ \hline \hline
           \multicolumn{10}{c}{Algorithm \ref{algo:multigrid1}}\\
           \hline
L=1                           &              &               &              &               &          &          & 5264         & 95.48 &      95.48              \\ \hline
L=2                         &              &               &            &           & 3036           & 10.76         & 30          & 0.550  &   11.31          \\ \hline
L=3                      &            &          & 1753         & 1.971      & 95        & 0.339        & 29      & 0.541  &     2.851        \\ \hline
L=4                   & 1002          & 0.456       & 242         & 0.291      & 105        & 0.395       & 31          &  0.572 &   1.714       \\ \hline
\hline
\multicolumn{10}{c}{Algorithm \ref{algo:multigrid2}}\\
           \hline
L=1                      &              &               &              &               &          &          & 100         & 2.375 &      2.375             \\ \hline
L=2                      &              &               &            &           & 99          & 1.224       & 2           & 0.113  &   1.337           \\ \hline
L=3                     &            &          & 99         & 0.157      & 4        & 0.095       & 2       & 0.121  &     0.373       \\ \hline
L=4                  & 100           & 0.024      & 8      & 0.019     & 4     & 0.097    & 2       & 0.121 &  0.261         \\ \hline
\end{tabular}
\end{table}

\paragraph{Algorithm \ref{algo:multigrid1}} 
If $L=1$, Algorithm \ref{algo:multigrid1} reduces to Algorithm \ref{algo:primaldual}, which takes $5264$ iterations to stop as Table \ref{tab:pd_comparel} shows. If $L=2$, we first conduct $3036$ iterations on a coarse grid $256\times256$. The obtained result is used to initialize the algorithm on the fine grid $512\times512$. With this initialization, Algorithm \ref{algo:multigrid1} only takes $30$ iterations to stop on the fine grid. Although extra calculations on $256\times256$ grid are required, the merit of fewer iterations on the fine grid overcomes the extra calculation cost. Thus, the total calculation time is reduced from $95.48$ seconds to $11.31$ seconds. When the number of levels $L$ gets even larger, the computing time could be further reduced. 
The results support Theorem \ref{lemma:cp}.
Since the  multilevel algorithm is able to dramatically reduce the calculation expense on the finest level, it consumes much less computing time.

\paragraph{Algorithm \ref{algo:multigrid2}} 
The advantage of Algorithm \ref{algo:multigrid2} is similar to that of Algorithm \ref{algo:multigrid1}. 
With  $L>1$, the number of iterations on the finest grid $512\times512$ can be reduced from $100$ to $2$, and the calculation time is reduced from $2.375$ seconds to $0.121$ seconds.

\subsection{The effect of stopping tolerances}
\label{sec:ada-tol}
In this subsection, we study tolerance rules for Algorithms \ref{algo:multigrid1} and \ref{algo:multigrid2}. With $\alpha=1,0,-1,-5$, formula (\ref{eq:tol_formula_propose}) gives decreasing tolerance, constant tolerance and increasing tolerance respectively.
For Algorithm \ref{algo:multigrid1}, $\varepsilon_L$ is chosen from
$\varepsilon_L \in \{ 1\times 10^{-6}\times 2^p:p\in\{0,-1,-2,\cdots,-10\} \}$.
For Algorithm \ref{algo:multigrid2}, $\varepsilon_L$ is chosen from
$ \varepsilon_L \in \{ 1\times 10^{-4}\times 2^p:p\in\{5,4,3,\cdots,-10\} \}$.
We report the best $\varepsilon_L$ and compare the three tolerance rules in Tables \ref{tab:1m} and \ref{tab:2m}. The ``best'' $\varepsilon_L$ takes the largest possible value that the algorithm meets the following condition:
\begin{equation}
    \label{eq:cond_f}
    |f(m^*_{h_l})-f(m^{K_l}_{h_l})| < |f(m^*_{h_l})-f(m^*_{h_{l-1}})|,\quad l = 2,3,\cdots,L.
\end{equation}
We use the above condition because (\ref{eq:lemma12}) and (\ref{eq:lemma42}) are intractable to numerically validate due to the non-uniqueness of $z^*_{h_l}$ and $y^*_{h_l}$. We use the error in objective function value $|f(m^*_{h_l})-f(m^{K_l}_{h_l})|$ to estimate $\|z^*_{h_l} - z^{K_l}_{h_l}\|^2$ and use the difference between the function values on two successive levels $|f(m^*_{h_l})-f(m^*_{h_{l-1}})|$ to estimate the grid error. Once condition (\ref{eq:cond_f}) is satisfied, we regard $m^{K_l}_{h_l}$ close enough to $m^*_{h_l}$.

\begin{table}[h]
\label{tab:1m}
\centering
\caption{Algorithm \ref{algo:multigrid1} on an example of size $256\times256$. 
``Error'' means $|f(m^*_{h_l})-f(m^{K_l}_{h_l})|$, ``Time'' is in seconds
. Baseline: Algorithm \ref{algo:primaldual} with  $\varepsilon = 1.6\times10^{-8}$, $33232$ iterations, time: $105.59$ seconds.}
\begin{tabular}{c|c|c|c|c|c|c}
\hline\hline
Level & $l=1$ & $l=2$ & $l=3$ & $l=4$ & $l=5$ & $l=6$ \\\hline
$h_l$    &  $h_l = 2^{-3}$    &   $h_l = 2^{-4}$      &    $h_l = 2^{-5}$      &     $h_l = 2^{-6}$     & $h_l = 2^{-7}$        &  $h_l = 2^{-8}$            \\\hline
\multicolumn{2}{c|}{$|f(m^*_{h_l})-f(m^*_{h_{l-1}})|$} & 1.1E-3 & 9.3e-4 & 5.8e-4 & 2.8e-4 & 1.3e-4 \\\hline
\hline
\multicolumn{7}{c}{$\varepsilon_l = 1\times 10^{-6}/512 \times (h_l/h_L) $}        \\\hline
$\varepsilon_l$      & 6.3E{-8} & 3.1E{-8} & 1.6E-8 & 7.8E-9 & 3.9E-9 & 2.0E-9 \\\hline
$K_l$     & 486     & 673     & 1526    & 3426    & 8401    & 12075   \\\hline
Time      &  0.055       &   0.065      &   0.204      &    1.121     &     8.494    & 39.22        \\\hline
Error     &    3.7E{-4}     &    2.8E{-4}     &   8.4E{-5}      &   2.6E{-5}      &    1.3E{-4}     &    6.4E{-5}    \\\hline
\multicolumn{7}{c}{Total time: $49.17$ secs. Speedup: $105.59/48.17=2.19$.}        \\\hline
\hline
\multicolumn{7}{c}{$\varepsilon_l = 1\times 10^{-6}/128 $}        \\\hline
$\varepsilon_l$       & 7.8E{-9} & 7.8E{-9} & 7.8E{-9} & 7.8E{-9} & 7.8E{-9} & 7.8E{-9} \\\hline
$K_l$   & 610     & 999     & 1921    & 3730    & 3806    & 4270    \\\hline
Time      &   0.063      &    0.092     &  0.254       &  1.617       &   3.483      &   13.46      \\\hline
Error     &   1.1E{-4}      &  1.9E{-4}      &    1.8E{-4}    &     6.8E{-5}    &    1.8E{-4}   &   1.4E{-5}    \\\hline
\multicolumn{7}{c}{Total time: $18.98$ secs. Speedup: $105.59/18.98=5.56$.}        \\\hline
\hline
 \multicolumn{7}{c}{$\varepsilon_l = 1\times 10^{-6}/32 \times (h_l/h_L)^{-1} $}       \\\hline
$\varepsilon_l$      & 9.8E{-10} & 2.0E{-9} & 3.9E{-9} & 7.8E{-9} & 1.6E{-8} & 3.1E{-8} \\\hline
$K_l$   & 735     & 1279     & 2351    & 4255    & 2459    & 932   \\\hline
Time      &  0.106       &   0.129     &   0.312      &  2.091	       &  2.914       &   	3.096    \\\hline
Error     &   2.3E{-5}      &  9.1E{-5}      &  6.5E{-5}      &  1.4E{-4}       &   2.4E{-4}      &  \textbf{2.3E{-6}}      \\\hline
\multicolumn{7}{c}{Total time: $8.68$ secs. Speedup: $105.59/8.68=\mathbf{12.16}$.}        \\\hline\hline
\multicolumn{7}{c}{$\varepsilon_l = 1\times 10^{-6}/8 \times (h_l/h_L)^{-5} $}       \\\hline
$\varepsilon_l$      & 3.7E{-15} & 1.2E{-13} & 3.8E{-12} & 1.2E{-10} & 3.9E{-9} & 1.3E{-7} \\\hline
$K_l$   & 1483     & 3354     & 7716    & 16532    & 7286    & \textbf{173 }   \\\hline
Time      &  0.731       &   0.300     &   0.998      &  7.072	       &  7.713       &   	\textbf{0.614}      \\\hline
Error     &   9.7E{-7}      &  5.7E{-7}      &  2.4E{-6}      &  1.6E{-5}       &   2.3E{-4}      &  {4.2E{-6}}      \\\hline
\multicolumn{7}{c}{Total time: $17.43$ secs. Speedup: $105.59/17.43={6.06}$.}        \\\hline\hline
\end{tabular}
\end{table}

Table \ref{tab:1m} shows the significance of initialization for Algorithm \ref{algo:multigrid1}. With $\alpha=1$, Algorithm \ref{algo:multigrid1} sets larger tolerance on the lower levels and the corresponding error is larger. Since we initialize level $l=6$ with the results on the lower levels $(1\leq l \leq 5)$, such tolerance rule leads to inaccurate initialization and larger number of iterations on level $l=6$. Calculations on the finest level are more computational complex than those on the coarser levels. Thus, the decreasing tolerance rule ($\alpha=1$) is not efficient. The increasing tolerances ($\alpha=-1,-5$) fix this issue. They invest more computations on the lower levels and takes less computations on the finest level. Then the total computing time is dramatically reduced.
However, if the $\alpha$ is too small, say $\alpha=-5$, the algorithm costs more overheads on the lower levels that hurt the total computing time, which validates Theorem \ref{theo:cp}.
Thus, $\alpha=-1$ is a good choice.

Table \ref{tab:2m} demonstrates similar properties of Algorithm \ref{algo:multigrid2}. With a decreasing tolerance rule ($\alpha=1$), the algorithm takes $42$ steps on level $l=6$, while this number for the increasing tolerances ($\alpha=-1,-5$) rule is only $1$. And $\alpha=-5$ over-invests calculations on lower levels and costs more time than $\alpha=-1$.

\begin{table}[h]
\label{tab:2m}
\centering
\caption{Algorithm \ref{algo:multigrid2} on an example of $512\times512$. ``Error'' means $|f(m^*_{h_l})-f(m^{K_l}_{h_l})|$, ``Time'' is in seconds, and the bold texts give the best results. Baseline: Algorithm \ref{algo:pdhg} with $\varepsilon=6.25\times10^{-6}$, $120$ iterations, time consumption: $2.595$ seconds.}
\begin{tabular}{c|c|c|c|c|c|c}
\hline\hline
Level & $l=1$ & $l=2$ & $l=3$ & $l=4$ & $l=5$ & $l=6$ \\\hline
$h_l$      &   $h_l = 2^{-4}$      &    $h_l = 2^{-5}$      &     $h_l = 2^{-6}$     & $h_l = 2^{-7}$        &  $h_l = 2^{-8}$    &  $h_l = 2^{-9}$          \\\hline
\multicolumn{2}{c|}{$|f(m^*_{h_l})-f(m^*_{h_{l-1}})|$} & 2.0E-3 & 1.0E-3 & 5.3E-4 & 2.4E-4 & 9.6E-5 \\\hline\hline
\multicolumn{7}{c}{$\varepsilon_l = 1\times 10^{-4}/16 \times (h_l/h_L) $}        \\\hline
$\varepsilon_l$       & 2.0E-4 & 1.0E-4 & 5.0E-5 & 2.5E-5 & 1.25E-5 & 6.25E-6 \\\hline
$K_l$    & 37     & 10    & 8    & 12    & 19    & 42   \\\hline
Time      &  0.001      &   0.001      &   0.003      &    0.020     &     0.249    & 0.955        \\\hline
Error     &    3.6E-4     &   5.5E-4    &   3.2E-4      &   3.0E-4     &    1.5E-4     &   8.9E-5    \\\hline
\multicolumn{7}{c}{Total time: $1.287$ secs. Speedup: $2.595/1.287=2.02$.}        \\\hline
\hline
\multicolumn{7}{c}{$\varepsilon_l = 1\times 10^{-4}/32 $}        \\\hline
$\varepsilon_l$      & 3.1E-6 & 3.1E-6 & 3.1E-6 & 3.1E-6 & 3.1E-6 & 3.1E-6 \\\hline
$K_l$    & 138     & 58     & 38    & 65    & 13    & 2    \\\hline
Time      &   0.002     &    0.003     &  0.009       &  0.089      &   0.181      &   0.113      \\\hline
Error     &   2.0E-5     & 4.5E-5      &    9.8E-5    &    7.4E-5    &   6.4E-5    &    6.2E-5    \\\hline
\multicolumn{7}{c}{Total time: $0.458$ secs. Speedup: $2.595/0.458=5.67$.}        \\\hline
 \hline
\multicolumn{7}{c}{$\varepsilon_l = 1\times 10^{-4}/32 \times (h_l/h_L)^{-1} $}        \\\hline
$\varepsilon_l$     & 9.8E-8 & 2.0E-7 & 3.9E-7 & 7.8E-7 & 1.6E-6 & 3.1E-6 \\\hline
$K_l$    & 352     & 547     & 22    & 69    & 2    &  \textbf{1}   \\\hline
Time      &  0.006      &   0.029      &   0.005      &  0.094	       &  0.058       & \textbf{0.068} 	      \\\hline
Error     &   1.5E-5     &  1.6E-5       &  4.6E-5      &  3.5E-5     &   5.3E-5     & \textbf{5.8E-5}      \\\hline
\multicolumn{7}{c}{Total time: $0.319$ secs. Speedup: $2.595/0.319=\mathbf{8.13}$.}        \\\hline\hline
\multicolumn{7}{c}{$\varepsilon_l = 1\times 10^{-4}\times 32 \times (h_l/h_L)^{-5} $}        \\\hline
$\varepsilon_l$     & 9.5E-11 & 3.1E-9 & 9.8E-8 & 3.1E-6 & 1.0E-4 & 3.2E-3 \\\hline
$K_l$    & 2433     & 1372     & 538    & 1    & 1    &  \textbf{1}   \\\hline
Time      &  0.037      &   0.075      &   0.113      &  0.005	       &  0.037       & \textbf{0.068} 	      \\\hline
Error     &   0     &  4.2E-17       &  5.5E-6      &  5.9E-5     &   7.4E-5     & {7.6E-5}      \\\hline
\multicolumn{7}{c}{Total time: $0.394$ secs. Speedup: $2.595/0.394={6.59}$.}        \\\hline\hline
\end{tabular}
\end{table}

\subsection{Comparison with other methods}
\label{sec:dotmark}
In this subsection, we compare our method with other EMD algorithms \cite{Ling,bassetti2018computation,li2018parallel,jacobs2018solving}. 
There are some other 2D EMD solvers~\cite{Pele,oberman2015efficient,solomon2015convolutional,benamou2015iterative,cuturi2013sinkhorn,haber_horesh_2015} we do not compare with. \cite{Pele} solves EMD with a thresholded metric; \cite{oberman2015efficient,haber_horesh_2015} are designed for Wasserstein-$p$ ($p>1$) distance; \cite{cuturi2013sinkhorn,solomon2015convolutional,benamou2015iterative}  solve EMD with the entropy regularizer, the objective function of which is not the same as ours. Thus, we are not able to compare these algorithms with ours fairly in our settings.

All the results are obtained on on the DOTmark~\cite{schrieber2017dotmark} dataset. We used $10$ images provided in the ``classic images'' of DOTmark.  Totally we calculated $45$ Wasserstein distances for all the $45$ image pairs. The time consumptions are averages taken on these image pairs. Figure \ref{fig:visualization_dot} in Appendix \ref{app:dotmark} visualizes two such images and the optimal transport between them. Tree-EMD~\cite{Ling} and Min-cost flow~\cite{bassetti2018computation} are exact algorithms stopping within finite steps. Other algorithms are iterative algorithms stopping by a tolerance. 
We take\footnote{The level number $L$ here is slightly smaller than the theoretical number $L = \log_2(1/h)+1$ but asymptotically the same.} $L = \log_2(1/h)-3$ and stopping tolerances as:
\begin{equation}
    \label{eq:stop_tol}
    \begin{aligned}
    \varepsilon =& \frac{1\times 10^{-6}}{512} \times \Big( \frac{h}{1/512} \Big)^{3}, & (\text{Alg. \ref{algo:primaldual}})\\
    \varepsilon_l =& \varepsilon_L \times \Big(\frac{h_l}{h}\Big)^{-1} ,~~\varepsilon_L = \frac{1\times 10^{-6}}{128} \times \Big( \frac{h}{1/512} \Big)^{2}, & (\text{Alg. \ref{algo:multigrid1}})\\
    \varepsilon =& \min\Big(2\times 10^{-4}, \frac{1\times 10^{-4}}{16} \times \Big( \frac{h}{1/512} \Big)^{2}\Big), & (\text{Alg. \ref{algo:pdhg}})\\
    \varepsilon_l =& \varepsilon_L \times \Big(\frac{h_l}{h}\Big)^{-1},~~\varepsilon = \min\Big(2\times 10^{-4}, \frac{1\times 10^{-4}}{32} \times \Big( \frac{h}{1/512} \Big)^{2}\Big). & (\text{Alg. \ref{algo:multigrid2}})
    \end{aligned}
\end{equation}
These formulas are summarized from experiment results. With these stopping tolerances, one can usually obtain solutions accurate enough.

\begin{table}[t]
\label{tab:err}
\centering
\caption{Errors of the algorithms with tolerance rule (\ref{eq:stop_tol}). ``Cal. Err'' means the calculation error: $|f(m^*_{h})-f(m^{K_L}_{h_L})|$. With (\ref{eq:stop_tol}), the calculation errors are comparable with the grid errors.
}
\begin{tabular}{c|l|c|c|c|c|c}
\hline
\multicolumn{2}{c|}{Grid size} & $32\times32$ & $64\times64$ & $128\times128$  & $256\times256$ & $512\times512$ \\ \hline
\hline
\multicolumn{7}{c}{$p=1$}                       \\\hline
\multicolumn{3}{c|}{Grid Error $|f(m^*_{2h})-f(m^*_{h})|$} & 1.15E-03 & 5.67E-04 & 2.82E-04 & 1.40E-04 \\\hline\hline
\multirow{4}{*}{Cal. Err.} & Alg. \ref{algo:primaldual} & 2.02E-03 & 1.06E-03 & 4.80E-04 & 2.90E-04  & N/A \\ \cline{2-7}
 & Alg. \ref{algo:multigrid1} & 2.02E-03 & 1.06E-03 & 5.20E-04 & 3.08E-04 & 1.71E-04\\\cline{2-7}
 & Alg. \ref{algo:pdhg} & 1.21E-03 & 1.03E-03 & 6.82E-04 & 2.70E-04 & 1.05E-04 \\ \cline{2-7}
 & Alg. \ref{algo:multigrid2} & 1.10E-03 & 7.59E-04 & 4.43E-04 & 2.11E-04  & 9.62E-05 \\ \hline\hline
\multicolumn{7}{c}{$p=2$}                       \\\hline
\multicolumn{3}{c|}{Grid Error $|f(m^*_{2h})-f(m^*_{h})|$} & 9.68E-04 & 4.74E-04 & 2.34E-04 & 1.16E-04 \\\hline\hline
\multirow{4}{*}{Cal. Err.} & Alg. \ref{algo:primaldual} & 1.83E-03 & 8.17E-04 & 4.63E-04 & 2.20E-04 & N/A \\\cline{2-7}
 & Alg. \ref{algo:multigrid1} & 1.74E-03 & 9.45E-04 & 4.47E-04 & 2.59E-04 & 1.32E-04      \\\cline{2-7}
 & Alg. \ref{algo:pdhg} & 8.95E-04 & 7.26E-04 & 4.74E-04 & 2.14E-04 & 9.83E-05 \\\cline{2-7}
 & Alg. \ref{algo:multigrid2} & 1.26E-03 & 1.13E-03 & 5.84E-04 & 2.58E-04 & 1.11E-04       \\ \hline\hline
\multicolumn{7}{c}{$p=\infty$}                       \\\hline
\multicolumn{3}{c|}{Grid Error $|f(m^*_{2h})-f(m^*_{h})|$} & 9.92E-04 & 5.14E-04 & 2.65E-04 & 1.36E-04 \\\hline\hline
\multirow{4}{*}{Cal. Err.} & Alg. \ref{algo:primaldual} & 1.64E-03 & 1.02E-03 & 4.26E-04 & 2.23E-04 & N/A \\\cline{2-7}
 & Alg. \ref{algo:multigrid1} & 1.60E-03 & 9.34E-04 & 4.65E-04 & 2.59E-04 & 1.09E-04   \\\cline{2-7}
 & Alg. \ref{algo:pdhg} & 1.19E-03 & 9.16E-04 & 5.54E-04 & 2.51E-04 & 1.09E-04 \\\cline{2-7}
 & Alg. \ref{algo:multigrid2} & 1.15E-03 & 9.82E-04 & 4.66E-04 & 2.57E-04 & 1.33E-04         \\ \hline
\end{tabular}
\end{table}

Table \ref{tab:err} reports the averaged errors on the examples of different sizes. The calculation errors is in the same order the grid approximation error:
\[|f(m^*_{h})-f(m^{K_L}_{h_L})| \leq 1.5 |f(m^*_{2h})-f(m^*_{h})|.\]Thus, formulas given in (\ref{eq:stop_tol}) are practically fair.

Table \ref{tab:final} reports the time consumptions of all the algorithms. ``N/A'' means the average calculation time is larger than half an hour.
On large-scale examples, Algorithm \ref{algo:multigrid1} is much faster than Algorithm \ref{algo:primaldual} and Algorithm \ref{algo:multigrid2} is much faster than Algorithm \ref{algo:pdhg}.
In most of the cases, Algorithm \ref{algo:multigrid2} is the best and Algorithms \ref{algo:multigrid1} and \ref{algo:pdhg} are also competitive. All the first-order methods (Algorithms \ref{algo:primaldual}, \ref{algo:pdhg}, \ref{algo:multigrid1}, \ref{algo:multigrid2}) are robust to the parameter $p$ in the ground metric $\|\cdot\|_p$. Tree-EMD~\cite{Ling} only works for $p=1$; the algorithm in \cite{bassetti2018computation} works well when $p=1,\infty$ and the grid size is not very large. As $p=2$, the algorithm in \cite{bassetti2018computation} requires large amount of memory and calculation time.

\begin{table}[h]
\centering
\caption{Averaged time consumption (seconds) on the DOTmark~\cite{schrieber2017dotmark} dataset. 
}
\label{tab:final}
\begin{tabular}{l|c|c|c|c|c}
\hline
Grid size & $32\times32$ & $64\times64$ & $128\times128$  & $256\times256$ & $512\times512$ \\ \hline\hline
\multicolumn{6}{c}{$p=1$}                       \\\hline
Tree-EMD~\cite{Ling} & 0.006 & 0.127 & 2.433 & 121.2 & N/A \\ \hline
Min-cost flow~\cite{bassetti2018computation} & {0.002} & 0.024 & 0.342 & 7.164 & 157.7 \\ \hline
Algorithm \ref{algo:primaldual}~\cite{li2018parallel} & 0.058 & 0.732 & 6.797 & 66.801 & N/A \\ \hline
Algorithm \ref{algo:multigrid1} & 0.047 & 0.219 & 1.321 & 7.555 & 49.054  \\ \hline
Algorithm \ref{algo:pdhg}~\cite{jacobs2018solving} & {0.002} & 0.008 & 0.066 & 0.980 & 3.102  \\ \hline
Algorithm \ref{algo:multigrid2} & $\mathbf{ 0.001}$ & $\mathbf{0.003}$ & $\mathbf{0.010}$ & $\mathbf{0.060}$ & $\mathbf{0.227}$ \\ \hline\hline
\multicolumn{6}{c}{$p=2$}                       \\\hline
Tree-EMD~\cite{Ling} & N/A  &N/A  &N/A  &N/A  &N/A  \\\hline
Min-cost flow~\cite{bassetti2018computation} & 0.082 & 1.863 & N/A &N/A &N/A \\\hline
Algorithm \ref{algo:primaldual}~\cite{li2018parallel} & 0.045 & 0.440 & 3.306 & 30.890 & N/A \\\hline
Algorithm \ref{algo:multigrid1} & 0.040 & 0.170 & 1.137 & 5.663 & 45.333       \\\hline
Algorithm \ref{algo:pdhg}~\cite{jacobs2018solving} & $\mathbf{0.002}$ & 0.008 & 0.061 & 0.809 & 2.471 \\\hline
Algorithm \ref{algo:multigrid2} & $\mathbf{0.002}$ & $\mathbf{0.004}$ & $\mathbf{0.018}$ & $\mathbf{0.163}$ & $\mathbf{0.826}$       \\ \hline\hline
\multicolumn{6}{c}{$p=\infty$}                       \\\hline
Tree-EMD~\cite{Ling} & N/A  &N/A  &N/A  &N/A  &N/A  \\\hline
Min-cost flow~\cite{bassetti2018computation} & $\mathbf{0.002}$  & 0.025  & 0.300 & 5.380 & 118.1 \\ \hline
Algorithm \ref{algo:primaldual}~\cite{li2018parallel} & 0.126 & 1.152 & 9.221 & 92.016 & N/A \\\hline
Algorithm \ref{algo:multigrid1} & 0.076 & 0.322 & 1.919 & 11.124 & 83.981  \\\hline
Algorithm \ref{algo:pdhg}~\cite{jacobs2018solving} & $\mathbf{0.002}$ & 0.007 & 0.053 & 0.707 & 2.257 \\\hline
Algorithm \ref{algo:multigrid2} & $\mathbf{0.002}$ & $\mathbf{0.004}$ & $\mathbf{0.025}$ & $\mathbf{0.268}$ & $\mathbf{1.416}$         \\ \hline
\end{tabular}
\end{table}

\section{Conclusion}

In this paper, we have proposed two multilevel algorithms for the computation of the Wasserstein-1 metric. The algorithms
leverage the $L_1$ type primal-dual structure in the minimal flux formulation of optimal transport. The multilevel setting provides very good initializations for the minimization problems on the fine grids. So it can significantly reduce the number of iterations on the finest grid. This consideration allows us to compute the metric between two $512\times 512$ images in about one second on a single CPU.
It is worth mentioning that the proposed algorithm also provides the Kantorovich potential and the optimal flux function between two densities. They are useful for the related Wasserstein variation problems \cite{Peyre}.

In future work, we will apply the multilevel method to optimal transport related minimization in mean field games and machine learning.

\newpage

 \appendix
\section{The proof of Lemma \ref{lemma:inter_bdd}}
\label{app:inter}

\begin{proof} First, we consider the interpolation of potential $\phi_{h_{l-1}}$. With $\phi_{h_l} =$ Interpolate $(\phi_{h_{l-1}})$, we have
\[
\begin{aligned}
&\|\phi_{h_l}\|_{L^2}^2 = \sum_{x \in \Omega^{h_l}} \phi_{h_l}^2(x) (h_l)^d\\
=& \sum_{x \in \Omega^{h_l}} \bigg(\frac{1}{2^{|\bar{J}|}}  \sum_{|y_{j_1} - x_{j_1}| \leq h_l} \cdots \sum_{|y_{j_{|\bar{J}|}} - x_{j_{|\bar{J}|}}| \leq h_l} \phi_{h_{l-1}}(y_{J}, y_{j_1},y_{j_2},\cdots, y_{j_{|\bar{J}|}}) \bigg)^2 (h_l)^d \\
\leq & \sum_{x \in \Omega^{h_l}} \bigg(\frac{1}{2^{|\bar{J}|}}  \sum_{|y_{j_1} - x_{j_1}| \leq h_l} \cdots \sum_{|y_{j_{|\bar{J}|}} - x_{j_{|\bar{J}|}}| \leq h_l} \phi^2_{h_{l-1}}(y_{J}, y_{j_1},y_{j_2},\cdots, y_{j_{|\bar{J}|}}) \bigg) (h_l)^d\\
= & \sum_{y \in \Omega^{h_{l-1}}} c(y) \phi^2_{h_{l-1}}(y) (h_l)^d.
\end{aligned}
  \]
The inequality in the third line above follows from Jensen's Inequality~\cite{boyd2004convex}, and $c(y)$ is a constant which can be bounded in the following way.

$\phi_{h_{l-1}}(y)$ contributes to the nodes within a $d$ dimensional hypercube $H(y) = \{x \in \Omega: \max_{j} |x_j - y_j| \leq h_l \}$. First, we consider $y$ as an interior point in $\Omega$. There are $2^d$ vertices in $H(y)$, for each vertex, the weight is $1/(2^d)$. There are $2^{d-1}d$ edges in $H(y)$, each edge contains a single point $x \in \Omega^{h_{l}}$, $\phi_{h_{l-1}}(y)$ contributes to $\phi_{h_l}(x)$ with weight $1/(2^{d-1})$. Generally speaking, there are $2^{d-n}\binom{d}{n}$ $n$-dimensional hypercubes on the boundary of $H(y)$~\cite{coxeter1973regular},
each $m$-dimensional hypercube contains a single point $x \in \Omega^{h_{l}}$, $\phi_{h_{l-1}}(y)$ contributes to $\phi_{h_l}(x)$ with weight $1/(2^{d-n})$. Moreover, the center of $H(y)$ is $y$, which is also in $\Omega^{h_l}$, $\phi_{h_{l-1}}(y)$ contributes to $\phi_{h_l}(y)$ with weight $1$. Thus, for interior point $y$, we have
\[
\begin{aligned}
c(y) =& \frac{1}{2^d} \times 2^d + \frac{1}{2^{d-1}} \times d 2^{d-1} + \cdots + \frac{1}{2^{d-1}} \times 2^{d-1}\binom{d}{n} + \cdots + 1 \times 1 \\
=& \sum_{n=0}^d \binom{d}{n} = (1+1)^d = 2^d.
\end{aligned}
 \]
For $y$ on the boundary of $\Omega$, there are less points in $H(y) \cap  \Omega^{h_l}$, and the weight for each node is the same as above, thus, $c(y) < 2^d$. In one word, $c(y) \leq 2^d$ for all $y \in \Omega^{h_{l-1}}$. Then, we obtain
\[
\begin{aligned}
\|\phi_{h_l}\|_{L^2}^2 \leq &  \sum_{y \in \Omega^{h_{l-1}}} c(y) \phi^2_{h_{l-1}}(y) (h_l)^d \leq \sum_{y \in \Omega^{h_{l-1}}} \phi^2_{h_{l-1}}(y) (2 h_l)^d\\
=& \sum_{y \in \Omega^{h_{l-1}}} \phi^2_{h_{l-1}}(y) (h_{l-1})^d = \|\phi_{h_{l-1}}\|_{L^2}^2.
\end{aligned}
\]
With the same proof line, the interpolation of $m_{h_{l-1}}$ and $\varphi_{h_{l-1}}$ can also be proved. Inequalities (\ref{eq:interpolate_bdd}) are proved.
\end{proof}

\section{Kantorovich potential}
\label{app:potential}

The Kantorovich potential can be obtained directly from the dual solution of (\ref{eq:minmax_cp}). Thus, Algorithm \ref{algo:primaldual} directly gives the potential $\phi_h^*:\Omega^h\to\Re$. Algorithm \ref{algo:pdhg} solves (\ref{eq:minmax_pdhg}) and gives the gradient of $\phi_h^*$: $\varphi^*_h:\Omega^h\to\Re^d$.  We obtain $\phi_h^*$ given $\varphi^*_h$ by solving
\[
    A_h A^*_h \phi_h = A_h \varphi^*_h.
\]
The boundary condition is given by (\ref{eq:opt_flux}). Thus, the Laplacian operator $A_h A^*_h$ is invertible, where the solution is unique up to a constant shrift.
And $\phi_h^*$ is given by
\begin{equation}
    \label{eq:potential_pdhg}
     \phi_h^* = (A_h A^*_h)^{-1} A_h \varphi^*_h.
\end{equation}
The inverse Laplacian operator can be calculated efficiently by the FFT~\cite{jacobs2018solving}.

\section{Visualization of the cat example}
\label{app:cats}

The cat example is used in Section \ref{sec:cats}. We visualize the two distributions $\rho^0,\rho^1$ and the optimal transport between them in this section.
Figure \ref{fig:cp} visualizes the primal-dual pair $(m^*,\phi^*)$ obtained by Algorithm \ref{algo:multigrid1}, Figure \ref{fig:pdhg} visualizes the primal-dual pair $(m^*,\varphi^*)$ obtained by Algorithm \ref{algo:multigrid2}. 

\begin{figure}
\centering
\begin{tabular}{ccc}
\hspace{-10mm}
\subfigure[][{$\rho^0$}]{
\includegraphics[width=0.35\textwidth]{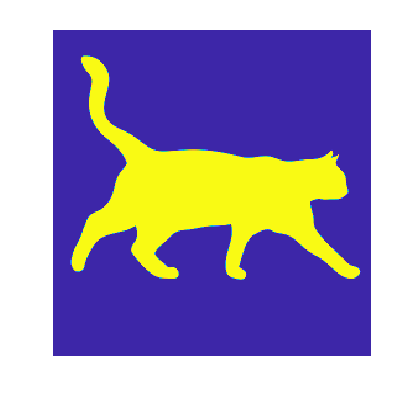}\label{fig:cat0}}
&
\hspace{-10mm}
\subfigure[][{$\rho^1$}]{
\includegraphics[width=0.35\textwidth]{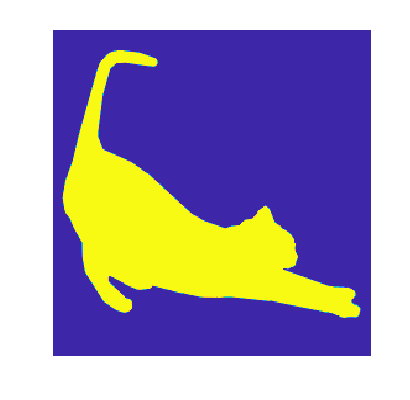}\label{fig:cat1}}\\
\hspace{-10mm}
\subfigure[][{$m^* (p=1)$}]{
\includegraphics[width=0.35\textwidth]{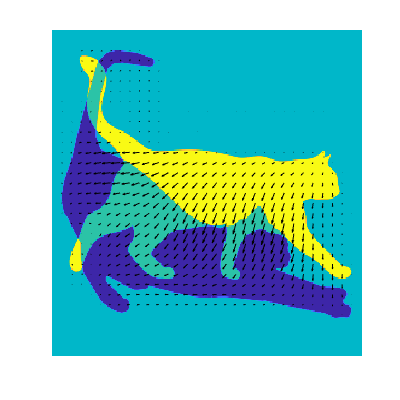}\label{fig:cpl1256}}
&
\hspace{-10mm}
\subfigure[][{$m^* (p=2)$}]{
\includegraphics[width=0.35\textwidth]{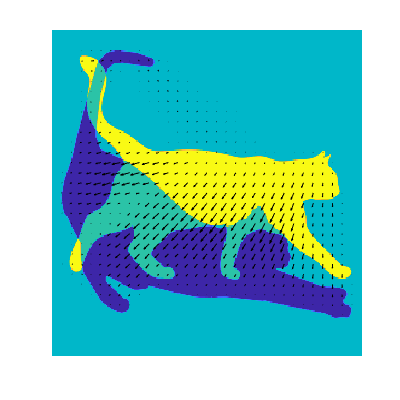}\label{fig:cpl2256}}
&
\hspace{-10mm}
\subfigure[][{$m^* (p=\infty)$}]{
\includegraphics[width=0.35\textwidth]{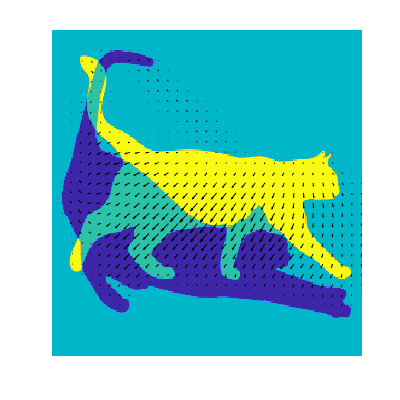}\label{fig:cpl256}}\\
\hspace{-10mm}
\subfigure[][{$\phi^* (p=1)$}]{
\includegraphics[width=0.35\textwidth]{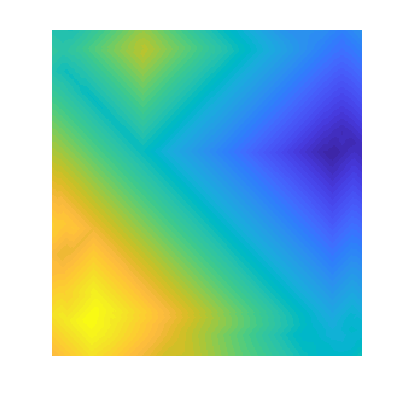}\label{fig:cpl1256p}}
&
\hspace{-10mm}
\subfigure[][{$\phi^* (p=2)$}]{
\includegraphics[width=0.35\textwidth]{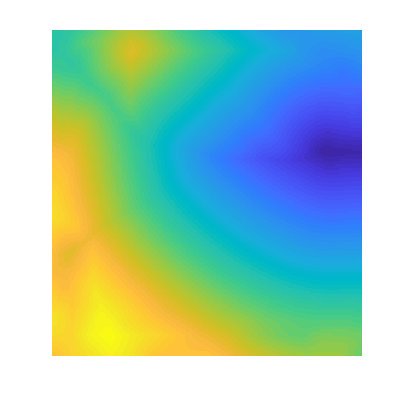}\label{fig:cpl2256p}}
&
\hspace{-10mm}
\subfigure[][{$\phi^* (p=\infty)$}]{
\includegraphics[width=0.35\textwidth]{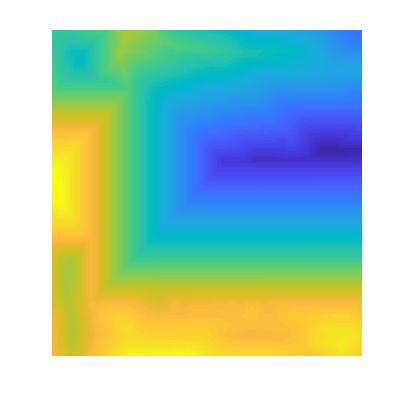}\label{fig:cpl3256p}}
\end{tabular}
\caption{Visualization of $(m^*,\phi^*)$ obtained by Algorithm \ref{algo:multigrid1}. $m^*$ is the optimal flux; $\phi^*$ is the Kantorovich potential.
 The backgrounds of Fig. \ref{fig:cpl1256}, \ref{fig:cpl2256} and \ref{fig:cpl256} are all $\rho^0-\rho^1$.}
\label{fig:cp}
\end{figure}

\begin{figure}
\centering
\begin{tabular}{ccc}
\hspace{-10mm}
\subfigure[][{$m^* (p=1)$}]{
\includegraphics[width=0.35\textwidth]{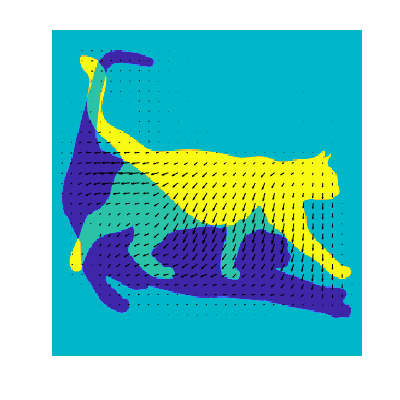}\label{fig:pdhgl1256}}
&
\hspace{-10mm}
\subfigure[][{$m^* (p=2)$}]{
\includegraphics[width=0.35\textwidth]{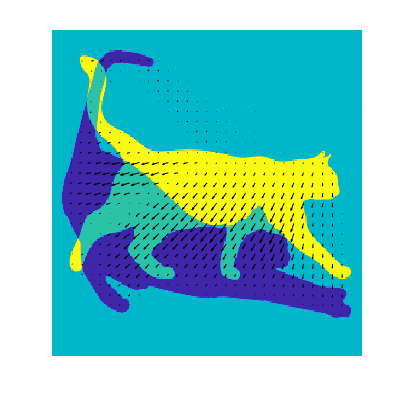}\label{fig:pdhgl2256}}
&
\hspace{-10mm}
\subfigure[][{$m^* (p=\infty)$}]{
\includegraphics[width=0.35\textwidth]{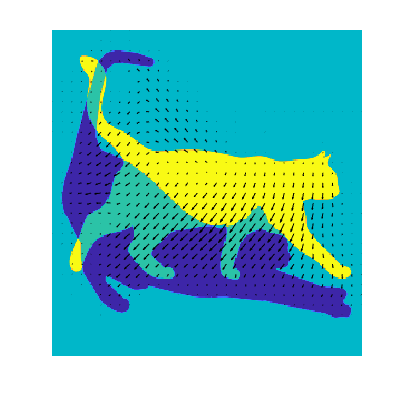}\label{fig:pdhgl256}}\\
\hspace{-10mm}
\subfigure[][{$\varphi^* (p=1)$}]{
\includegraphics[width=0.35\textwidth]{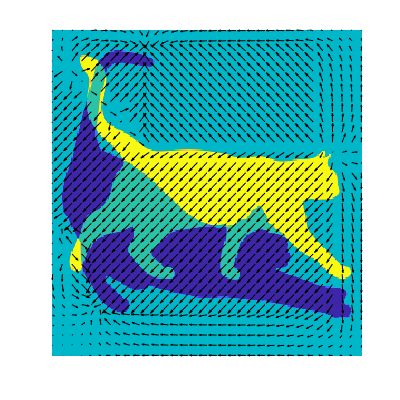}\label{fig:pdhgl1256p}}
&
\hspace{-10mm}
\subfigure[][{$\varphi^* (p=2)$}]{
\includegraphics[width=0.35\textwidth]{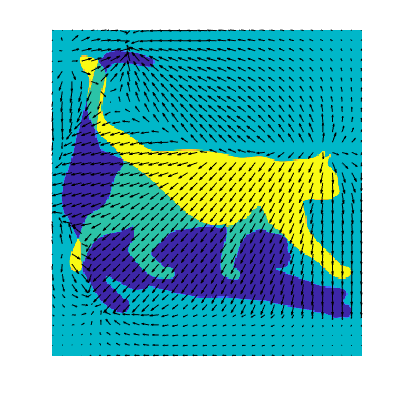}\label{fig:pdhgl2256p}}
&
\hspace{-10mm}
\subfigure[][{$\varphi^* (p=\infty)$}]{
\includegraphics[width=0.35\textwidth]{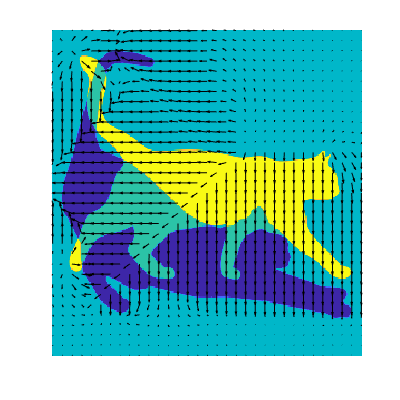}\label{fig:pdhg_l3256p}}\\
\hspace{-10mm}
\subfigure[][{$\phi^* (p=1)$}]{
\includegraphics[width=0.35\textwidth]{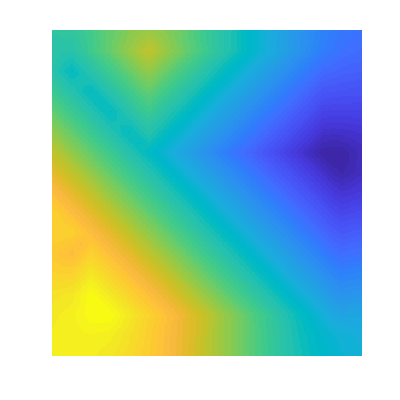}\label{fig:pdhgl1256phi}}
&
\hspace{-10mm}
\subfigure[][{$\phi^* (p=2)$}]{
\includegraphics[width=0.35\textwidth]{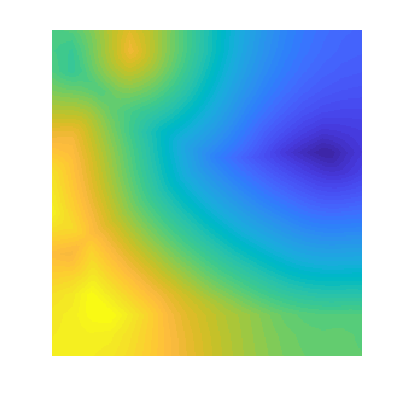}\label{fig:pdhgl2256phi}}
&
\hspace{-10mm}
\subfigure[][{$\phi^* (p=\infty)$}]{
\includegraphics[width=0.35\textwidth]{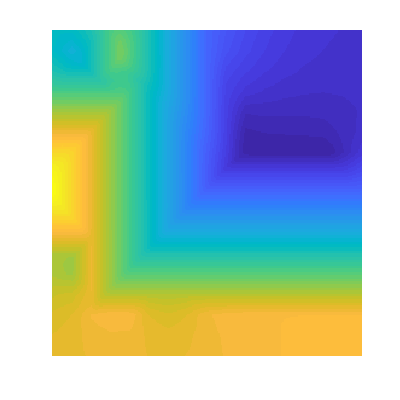}\label{fig:pdhg_l3256phi}}
\end{tabular}
\caption{Visualization of $(m^*,\varphi^*)$ and the potential $\phi^*$ obtained by Algorithm \ref{algo:multigrid2}. $\rho^0,\rho^1$ are the same with those in Figure \ref{fig:cp}. $m^*$ is the optimal flux; $\phi^*$ is the Kantorovich potential, that is obtained from $\varphi^*$ by the method in Appendix \ref{app:potential}.
}
\label{fig:pdhg}
\end{figure}

\section{Visualization of DOTmark}
\label{app:dotmark}

The DOTmark dataset is used in Section \ref{sec:dotmark}. We visualize two of the distributions $\rho^0,\rho^1$ and the optimal transport between them in Figure \ref{fig:visualization_dot}.

\begin{figure}
\centering
\begin{tabular}{ccc}
\hspace{-10mm}
\subfigure[][\parbox{0.25\textwidth}{$\rho^0$: Cameraman ($512\times512$), one image of DOTmark}]{
\includegraphics[width=0.35\textwidth]{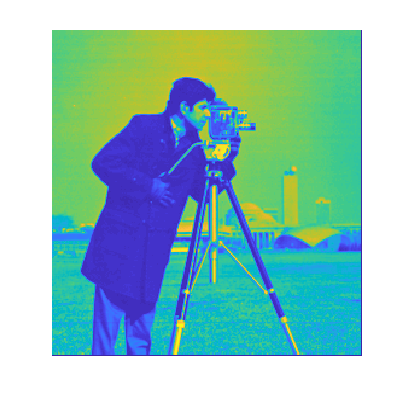}\label{fig:cameraman}}
&
\hspace{-10mm}
\subfigure[][\parbox{0.25\textwidth}{$\rho^1$: Lake ($512\times512$), one image of DOTmark}]{
\includegraphics[width=0.35\textwidth]{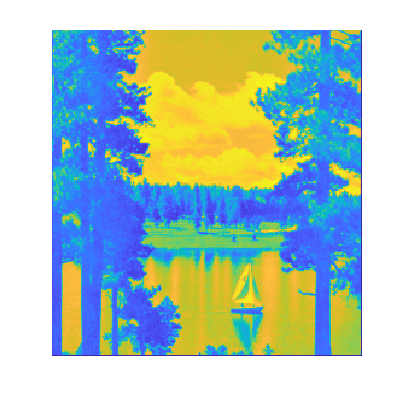}\label{fig:lena}}\\
\hspace{-10mm}
\subfigure[][{$p=1$}]{
\includegraphics[width=0.35\textwidth]{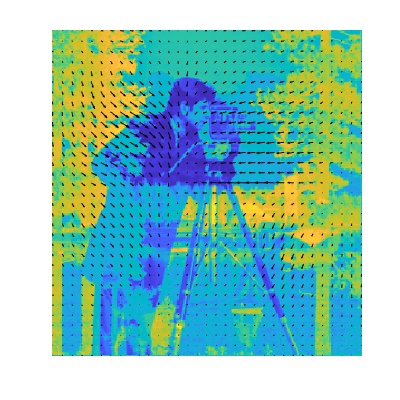}\label{fig:cameramanp1}}
&
\hspace{-10mm}
\subfigure[][{$p=2$}]{
\includegraphics[width=0.35\textwidth]{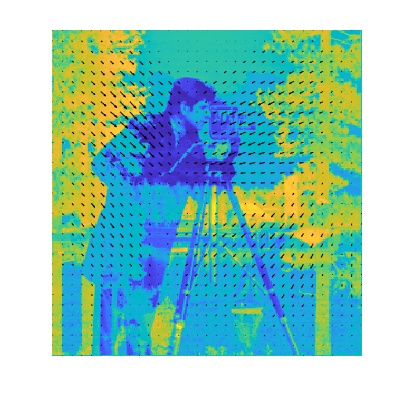}\label{fig:cameramanp2}}
&
\hspace{-10mm}
\subfigure[][{$p=\infty$}]{
\includegraphics[width=0.35\textwidth]{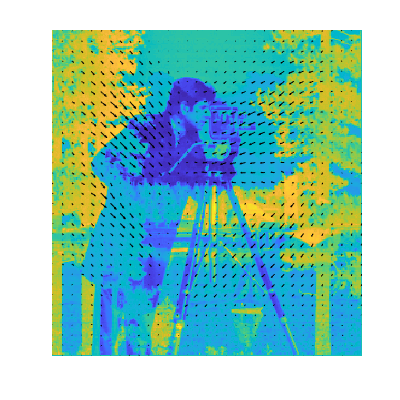}\label{fig:cameramanpinf}}
\end{tabular}
\caption{Visualization of the optimal transport $m^*$ between the two images $\rho^0$ and $\rho^1$. The background is the difference between $\rho^0,\rho^1$: $\rho=\rho^0-\rho^1$.
}
\label{fig:visualization_dot}
\end{figure}


\bibliographystyle{siamplain}
\bibliography{references}

\end{document}

%% file: ex_shared.tex
\usepackage{lipsum}
\usepackage{amsfonts}
\usepackage{graphicx}
\usepackage{epstopdf}
\usepackage{algorithmic}
\ifpdf
  \DeclareGraphicsExtensions{.eps,.pdf,.png,.jpg}
\else
  \DeclareGraphicsExtensions{.eps}
\fi

\usepackage{xr-hyper}
\usepackage{hyperref}
\usepackage{cleveref}

\usepackage{amsmath,amssymb}
\usepackage[english]{babel}
\usepackage[utf8x]{inputenc}
\usepackage[T1]{fontenc}
\usepackage[title]{appendix}
\usepackage{float}
\usepackage{graphicx}
\usepackage{epstopdf}
\usepackage{comment}
\usepackage{subfigure}
\usepackage[colorinlistoftodos]{todonotes}
\usepackage[export]{adjustbox}
\usepackage[algoruled,boxed,lined,algo2e]{algorithm2e}
\usepackage{multirow}
\usepackage{tikz}
\usepackage{enumerate}

\DeclareMathOperator*{\argmin}{arg\,min}
\DeclareMathOperator*{\argmax}{arg\,max}
\DeclareMathOperator*{\minimize}{minimize}
\DeclareMathOperator*{\maximize}{maximize}

\def \I  {\mathcal{I}}

\newtheorem{assume}{Assumption}

\newsiamremark{remark}{Remark}
\newsiamremark{hypothesis}{Hypothesis}
\crefname{hypothesis}{Hypothesis}{Hypotheses}
\newsiamthm{claim}{Claim}

\headers{Multilevel Optimal Transport
}{J. Liu, W. Yin, W. Li and Y.T. Chow}

\title{Multilevel Optimal Transport: a Fast Approximation of Wasserstein-1 distances\thanks{The codes will be released to: \url{https://github.com/liujl11git/multilevelOT}.
\funding{The research is supported by AFOSR MURI FA9550-18-1-0502, ONR N000141712162 and NSF DMS-1720237.
The Titan Xp used for this research was donated by the NVIDIA Corporation.}}}

\author{Jialin Liu\thanks{Department of Mathematics, University of California, Los Angeles.   (liujl11@math.ucla.edu; wotaoyin@math.ucla.edu; wcli@math.ucla.edu)}
\and 
Wotao Yin\footnotemark[2]
\and 
Wuchen Li\footnotemark[2]
\and 
Yat Tin Chow\thanks{Department of Mathematics, University of California, Riverside.  (yattinc@ucr.edu)}
}

\usepackage{amsopn}

\makeatletter
\newcommand*{\addFileDependency}[1]{
  \typeout{(#1)}
  \@addtofilelist{#1}
  \IfFileExists{#1}{}{\typeout{No file #1.}}
}
\makeatother

\newcounter{protocol}
\makeatletter
\newenvironment{protocol}[1][ht]{%
  \let\c@algorithm\c@protocol

  \begin{algorithm}%
  }{\end{algorithm}
}
\makeatother

%% file: iter_phi.tikz
\tikzset{every picture/.style={line width=0.75pt}} 

\begin{tikzpicture}[x=0.75pt,y=0.75pt,yscale=-1,xscale=1]

\draw    (72.39, 33.34) rectangle (188.03, 161.34)   ;
\draw  [fill={rgb, 255:red, 0; green, 0; blue, 0 }  ,fill opacity=1 ]  (72.39, 33.34) circle [x radius= 5.17, y radius= 5.17]  ;
\draw  [fill={rgb, 255:red, 0; green, 0; blue, 0 }  ,fill opacity=1 ]  (72.39, 161.34) circle [x radius= 5.17, y radius= 5.17]  ;
\draw  [fill={rgb, 255:red, 0; green, 0; blue, 0 }  ,fill opacity=1 ]  (188.03, 33.34) circle [x radius= 5.17, y radius= 5.17]  ;
\draw  [fill={rgb, 255:red, 0; green, 0; blue, 0 }  ,fill opacity=1 ]  (188.03, 161.34) circle [x radius= 5.17, y radius= 5.17]  ;
\draw    (254.39, 33.34) rectangle (370.02, 161.34)   ;
\draw    (254.39,97.34) -- (370.02,97.34) ;

\draw    (312.21,33.34) -- (312.21,161.34) ;

\draw  [fill={rgb, 255:red, 0; green, 0; blue, 0 }  ,fill opacity=1 ]  (254.39, 97.34) circle [x radius= 5.17, y radius= 5.17]  ;
\draw  [fill={rgb, 255:red, 0; green, 0; blue, 0 }  ,fill opacity=1 ]  (312.21, 97.34) circle [x radius= 5.17, y radius= 5.17]  ;
\draw  [fill={rgb, 255:red, 0; green, 0; blue, 0 }  ,fill opacity=1 ]  (254.39, 33.34) circle [x radius= 5.17, y radius= 5.17]  ;
\draw  [fill={rgb, 255:red, 0; green, 0; blue, 0 }  ,fill opacity=1 ]  (254.39, 161.34) circle [x radius= 5.17, y radius= 5.17]  ;
\draw  [fill={rgb, 255:red, 0; green, 0; blue, 0 }  ,fill opacity=1 ]  (312.21, 33.34) circle [x radius= 5.17, y radius= 5.17]  ;
\draw  [fill={rgb, 255:red, 0; green, 0; blue, 0 }  ,fill opacity=1 ]  (370.02, 33.34) circle [x radius= 5.17, y radius= 5.17]  ;
\draw  [fill={rgb, 255:red, 0; green, 0; blue, 0 }  ,fill opacity=1 ]  (370.02, 97.34) circle [x radius= 5.17, y radius= 5.17]  ;
\draw  [fill={rgb, 255:red, 0; green, 0; blue, 0 }  ,fill opacity=1 ]  (312.21, 161.34) circle [x radius= 5.17, y radius= 5.17]  ;
\draw  [fill={rgb, 255:red, 0; green, 0; blue, 0 }  ,fill opacity=1 ]  (370.02, 161.34) circle [x radius= 5.17, y radius= 5.17]  ;

\draw (365,85) node [scale=0.9]   {$\phi_{0.5}(0.5, 1)$};
\draw (310,20) node [scale=0.9] {$\phi_{0.5}(0, 0.5)$};
\draw (310,175) node [scale=0.9]   {$\phi_{0.5}(1, 0.5)$};
\draw (250,85) node [scale=0.9]  {$\phi_{0.5}(0.5,0)$};
\draw (310,110) node [scale=0.9]   {$\phi_{0.5}(0.5, 0.5)$};
\draw (90,20) node [scale=0.9] {$\phi_1(0,0)$};
\draw (250,20) node [scale=0.9]  {$\phi_1(0,0)$};
\draw (195,20) node [scale=0.9]  {$\phi_1(0,1)$};
\draw (375,20) node [scale=0.9]  {$\phi_1(0,1)$};
\draw (90,175) node [scale=0.9]  {$\phi_1(1,0)$};
\draw (250,175) node [scale=0.9]  {$\phi_1(1,0)$};
\draw (195,175) node [scale=0.9] {$\phi_1(1,1)$};
\draw (375,175) node [scale=0.9]  {$\phi_1(1,1)$};

\draw (330,200) node [scale=0.9]  {$\phi_{0.5} = \mathrm{Interpolate }(\phi_1), h = 0.5$};
\draw (135,200) node [scale=0.9]  {$\phi_{1}, h = 1$};

\draw (200,290) node [scale=0.9]  { $\phi_{0.5}(0.5, 0.5) = \frac{\phi_1(0,0)+\phi_1(0,1)+\phi_1(1,0)+\phi_1(1,1)}{4}$};

\draw (150,230) node [scale=0.9]  {$\phi_{0.5}(0.5, 1) =\frac{\phi_1(0,1)+\phi_1(1,1)}{2}$};
\draw (310,230) node [scale=0.9] {$\phi_{0.5}(0, 0.5) =\frac{\phi_1(0,0)+\phi_1(0,1)}{2}$};
\draw (150,260) node [scale=0.9]  {$\phi_{0.5}(1, 0.5) =\frac{\phi_1(1,0)+\phi_1(1,1)}{2}$};
\draw (310,260) node [scale=0.9]  {$\phi_{0.5}(0.5, 0) =\frac{\phi_1(0,0)+\phi_1(1,0)}{2}$};

\end{tikzpicture}

%% file: iter_m.tikz
\tikzset{every picture/.style={line width=0.75pt}} 

\begin{tikzpicture}[x=0.75pt,y=0.75pt,yscale=-1,xscale=1]

\draw    (102.39, 33.34) rectangle (218.03, 161.34)   ;
\draw  [fill={rgb, 255:red, 0; green, 0; blue, 0 }  ,fill opacity=1 ]  (102.39, 33.34) circle [x radius= 5.17, y radius= 5.17]  ;
\draw  [fill={rgb, 255:red, 0; green, 0; blue, 0 }  ,fill opacity=1 ]  (102.39, 161.34) circle [x radius= 5.17, y radius= 5.17]  ;
\draw  [fill={rgb, 255:red, 0; green, 0; blue, 0 }  ,fill opacity=1 ]  (218.03, 33.34) circle [x radius= 5.17, y radius= 5.17]  ;
\draw  [fill={rgb, 255:red, 0; green, 0; blue, 0 }  ,fill opacity=1 ]  (218.03, 161.34) circle [x radius= 5.17, y radius= 5.17]  ;
\draw    (254.39, 33.34) rectangle (370.02, 161.34)   ;
\draw    (254.39,97.34) -- (370.02,97.34) ;

\draw    (312.21,33.34) -- (312.21,161.34) ;

\draw  [fill={rgb, 255:red, 0; green, 0; blue, 0 }  ,fill opacity=1 ]  (254.39, 97.34) circle [x radius= 5.17, y radius= 5.17]  ;
\draw  [fill={rgb, 255:red, 0; green, 0; blue, 0 }  ,fill opacity=1 ]  (312.21, 97.34) circle [x radius= 5.17, y radius= 5.17]  ;
\draw  [fill={rgb, 255:red, 0; green, 0; blue, 0 }  ,fill opacity=1 ]  (254.39, 33.34) circle [x radius= 5.17, y radius= 5.17]  ;
\draw  [fill={rgb, 255:red, 0; green, 0; blue, 0 }  ,fill opacity=1 ]  (254.39, 161.34) circle [x radius= 5.17, y radius= 5.17]  ;
\draw  [fill={rgb, 255:red, 0; green, 0; blue, 0 }  ,fill opacity=1 ]  (312.21, 33.34) circle [x radius= 5.17, y radius= 5.17]  ;
\draw  [fill={rgb, 255:red, 0; green, 0; blue, 0 }  ,fill opacity=1 ]  (370.02, 33.34) circle [x radius= 5.17, y radius= 5.17]  ;
\draw  [fill={rgb, 255:red, 0; green, 0; blue, 0 }  ,fill opacity=1 ]  (370.02, 97.34) circle [x radius= 5.17, y radius= 5.17]  ;
\draw  [fill={rgb, 255:red, 0; green, 0; blue, 0 }  ,fill opacity=1 ]  (312.21, 161.34) circle [x radius= 5.17, y radius= 5.17]  ;
\draw  [fill={rgb, 255:red, 0; green, 0; blue, 0 }  ,fill opacity=1 ]  (370.02, 161.34) circle [x radius= 5.17, y radius= 5.17]  ;

\draw (360,132) node [rotate=90,scale=0.9]   {$m_{h_2,1}(0.5, 1)$};
\draw (325,68) node [rotate=90,scale=0.9] {$m_{h_2,1}(0, 0.5)$};
\draw (245,140) node [rotate=90,scale=0.9]  {$m_{h_2,1}(0.5,0)$};
\draw (302,135) node [rotate=90,scale=0.9]   {$m_{h_2,1}(0.5, 0.5)$};

\draw (115,100) node [rotate=90,scale=0.9] {$m_{h_1,1}(0,0)$};
\draw (245,65) node [rotate=90,scale=0.9]  {$m_{h_2,1}(0,0)$};
\draw (200,100) node [rotate=90,scale=0.9]  {$m_{h_1,1}(0,1)$};
\draw (360,65) node [rotate=90,scale=0.9]  {$m_{h_2,1}(0,1)$};

\draw (330,190) node [scale=0.9]  {$m_{h_2,1} = \mathrm{Interpolate }(m_{h_1,1}), h_2 = 0.5$};
\draw (135,190) node [scale=0.9]  {$m_{h_1,1}, h_1 = 1$};

\draw (490,150) node [scale=0.9]  { $m_{h_2,1}(0.5, 0.5) = \frac{m_{h_1,1}(0,0)+m_{h_1,1}(0,1)}{2}$};

\draw (470,60) node [scale=0.9] {$m_{h_2,1}(0.5,0) =m_{h_1,1}(0,0)$};
\draw (470,90) node [scale=0.9]  {$m_{h_2,1}(0.5,1) =m_{h_1,1}(0,1)$};
\draw (490,120) node [scale=0.9]  {$m_{h_2,1}(0,0.5) =\frac{m_{h_1,1}(0,0)+m_{h_1,1}(0,1)}{2}$};

\end{tikzpicture}

%% file: iter_m2.tikz
\tikzset{every picture/.style={line width=0.75pt}} 

\begin{tikzpicture}[x=0.75pt,y=0.75pt,yscale=-1,xscale=1]

\draw    (102.39, 33.34) rectangle (218.03, 161.34)   ;
\draw  [fill={rgb, 255:red, 0; green, 0; blue, 0 }  ,fill opacity=1 ]  (102.39, 33.34) circle [x radius= 5.17, y radius= 5.17]  ;
\draw  [fill={rgb, 255:red, 0; green, 0; blue, 0 }  ,fill opacity=1 ]  (102.39, 161.34) circle [x radius= 5.17, y radius= 5.17]  ;
\draw  [fill={rgb, 255:red, 0; green, 0; blue, 0 }  ,fill opacity=1 ]  (218.03, 33.34) circle [x radius= 5.17, y radius= 5.17]  ;
\draw  [fill={rgb, 255:red, 0; green, 0; blue, 0 }  ,fill opacity=1 ]  (218.03, 161.34) circle [x radius= 5.17, y radius= 5.17]  ;
\draw    (254.39, 33.34) rectangle (370.02, 161.34)   ;
\draw    (254.39,97.34) -- (370.02,97.34) ;

\draw    (312.21,33.34) -- (312.21,161.34) ;

\draw  [fill={rgb, 255:red, 0; green, 0; blue, 0 }  ,fill opacity=1 ]  (254.39, 97.34) circle [x radius= 5.17, y radius= 5.17]  ;
\draw  [fill={rgb, 255:red, 0; green, 0; blue, 0 }  ,fill opacity=1 ]  (312.21, 97.34) circle [x radius= 5.17, y radius= 5.17]  ;
\draw  [fill={rgb, 255:red, 0; green, 0; blue, 0 }  ,fill opacity=1 ]  (254.39, 33.34) circle [x radius= 5.17, y radius= 5.17]  ;
\draw  [fill={rgb, 255:red, 0; green, 0; blue, 0 }  ,fill opacity=1 ]  (254.39, 161.34) circle [x radius= 5.17, y radius= 5.17]  ;
\draw  [fill={rgb, 255:red, 0; green, 0; blue, 0 }  ,fill opacity=1 ]  (312.21, 33.34) circle [x radius= 5.17, y radius= 5.17]  ;
\draw  [fill={rgb, 255:red, 0; green, 0; blue, 0 }  ,fill opacity=1 ]  (370.02, 33.34) circle [x radius= 5.17, y radius= 5.17]  ;
\draw  [fill={rgb, 255:red, 0; green, 0; blue, 0 }  ,fill opacity=1 ]  (370.02, 97.34) circle [x radius= 5.17, y radius= 5.17]  ;
\draw  [fill={rgb, 255:red, 0; green, 0; blue, 0 }  ,fill opacity=1 ]  (312.21, 161.34) circle [x radius= 5.17, y radius= 5.17]  ;
\draw  [fill={rgb, 255:red, 0; green, 0; blue, 0 }  ,fill opacity=1 ]  (370.02, 161.34) circle [x radius= 5.17, y radius= 5.17]  ;

\draw (345,45) node [scale=0.9] {$m_{h_2,2}(0, 0.5)$};
\draw (345,175) node [scale=0.9]   {$m_{h_2,2}(1, 0.5)$};
\draw (280,85) node [scale=0.9]  {$m_{h_2,2}(0.5,0)$};
\draw (345,110) node [scale=0.9]   {$m_{h_2,2}(0.5, 0.5)$};

\draw (160,45) node [scale=0.9] {$m_{h_1,2}(0,0)$};
\draw (280,45) node [scale=0.9]  {$m_{h_2,2}(0,0)$};
\draw (160,175) node [scale=0.9]  {$m_{h_1,2}(1,0)$};
\draw (280,150) node [scale=0.9]  {$m_{h_2,2}(1,0)$};

\draw (330,200) node [scale=0.9]  {$m_{h_2,2} = \mathrm{Interpolate }(m_{h_1,2}), h_2 = 0.5$};
\draw (135,200) node [scale=0.9]  {$m_{h_1,2}, h_1 = 1$};

\draw (490,150) node [scale=0.9]  { $m_{h_2,2}(0.5, 0.5) = \frac{m_{h_1,2}(0,0)+m_{h_1,2}(1,0)}{2}$};

\draw (470,60) node [scale=0.9] {$m_{h_2,2}(0, 0.5) =m_{h_1,2}(0,0)$};
\draw (470,90) node [scale=0.9]  {$m_{h_2,2}(1, 0.5) =m_{h_1,2}(1,0)$};
\draw (490,120) node [scale=0.9]  {$m_{h_2,2}(0.5, 0) =\frac{m_{h_1,2}(0,0)+m_{h_1,2}(1,0)}{2}$};

\end{tikzpicture}